\documentclass[final,3p,times,onecolumn,authoryear]{elsarticle}    

\usepackage{graphicx}          
\usepackage[dvips]{epsfig}    
 \usepackage{multicol} 
\usepackage{ulem}
\usepackage{color}
\usepackage{braket}
\usepackage{amsmath,mathrsfs,amssymb}
\allowdisplaybreaks[4]
\usepackage{subeqnarray}
\usepackage{cases}
\usepackage{ntheorem}
\usepackage{amssymb}

\usepackage{wrapfig}
\usepackage{flushend}

\newtheorem{Theorem}{Theorem}
\newtheorem{mypro}{Proposition}
\newtheorem*{proof}{Proof}

\newtheorem{assumption}{Assumption}

\newtheorem{remark}{Remark}

\begin{document}
\begin{frontmatter}

\title{On the non-Markovian quantum stochastic network dynamics}

\author[Haijin Ding]{Haijin Ding\corref{cor1}}\ead{dhj17@tsinghua.org.cn}

\author[Guofeng1,Guofeng Zhang]{Guofeng Zhang}\ead{guofeng.zhang@polyu.edu.hk}

\address[Haijin Ding]{Department of Mechanical and Automation Engineering, The Chinese University of Hong Kong, {Hong Kong}}

\address[Guofeng1]{Department of Applied Mathematics, The Hong Kong Polytechnic University, {Hung Hom}, {Kowloon}, {Hong Kong}}

\address[Guofeng Zhang]{The Hong Kong Polytechnic University Shenzhen Research Institute, {Shenzhen}, {518057}, {China}}

\cortext[cor1]{Corresponding author.}

\begin{abstract}
\linespread{1.5}
In this paper, we investigate non-Markovian quantum dynamics from the perspective of quantum noises in a network of atoms mediated by a waveguide. In such networks, quantum coherent feedback control becomes achievable when coherent fields (or quantum noises) in the format of photons with continuous modes propagate through the waveguide. Different from traditional Markovian quantum systems, the non-Markovian quantum network can be regarded as a quantum system interacting with multiple input quantum noise channels with different time delays. Then the \rm{It\={o}} relationships among different quantum noise channels are determined by the quantum noise commutators and rely on the distances among atoms as well as their coupling strengths to the waveguide.
The non-Markovian dynamics of such quantum networks can be modeled with the quantum stochastic differential equation (QSDE) containing integral kernels determined by the commutators among quantum noise operators. Utilizing this stochastic approach related to quantum noises, the filtering of quantum states can be modulated by parameters such as atom-waveguide coupling strengths and quantum control amplitudes.
\end{abstract}


\end{frontmatter}

\tableofcontents

\section{Introduction}
Non-Markovian dynamical processes can be distinguished from Markovian processes due to historical memory effects in both classical~\citep{mariton1989systems,strickland1989sensitivity} and quantum systems~\citep{diosi2012non,zhang2013non,xue2016feedback,ding2024non,saldi2024quantum}. For instance,  in quantum networks interconnected by waveguides with photons propagating among different nodes, the evolution of quantum systems can be represented with differential equations with time delays, illustrating that the dynamics of quantum states is influenced not only by its current state, but also by historical quantum states~\citep{nemet2019comparison,arranz2021modeling,ding2023quantum,DingNlevel2024}. Then the control properties, i.e., the steady states or the convergence rate, can be influenced by parameters such as the length of time delays~\citep{DingNlevel2024}. Besides, when a quantum system interacts with an environment made up of a group of oscillators, the interaction between the quantum system and environment can be modeled by stochastic integral processes with non-Markovian integral kernels~\citep{diosi2012non,ding2024non}. After averaging over the environmental noises, the decay rate of quantum systems to the environment can be time-varying rather than a constant in the simplified Markovian circumstnace~\citep{ding2024non}. Thus, both the above non-Markovian quantum dynamics with time delays and time-varying decay rates indicate that non-Markovian quantum control dynamics can be more complex than the Markovian scenario~\citep{In2017Dynamics,breuer2016colloquium}.

In both of the two examples above, the non-Markovian dynamics of the quantum system we are interested in arises from its interactions with the environment, which can be oscillators with finite modes~\citep{diosi2012non,diosi1998non,ding2024non}, a waveguide with infinite continuous modes~\citep{nemet2019comparison,arranz2021modeling,ding2023quantum,DingNlevel2024} or other more general circumstance with continuous spectrums~\citep{zhang2013non}.  By transforming the modeling of the environment to the time domain, the influence upon quantum states by the environment can be equivalently elucidated by input time-varying quantum noises~\citep{gardiner2004quantum}. Then the non-Markovian dynamics of quantum systems driven by input quantum noises can be evaluated by quantum stochastic differential equations (QSDE) containing an integral process with memory effects~\citep{zhang2013non}, which are generalized from the widely studied Markovian circumstances~\citep{james2008hTAC,dong2022dynamicsTAC,vuglar2016quantum,qi2012further,wang2016fault,van2005feedback,fischer2018particle}.

Apart from the modeling of the interactions between a quantum system and environment, the estimation and filtering of quantum states are widely used in quantum information processing and quantum engineering~\citep{SIAMFiltering,gao2016fault}. Similar to the control and filtering of Markovian quantum systems, the measurement information of quantum states in a non-Markovian system can be used to estimate the quantum state evolutions and realize measurement feedback controls~\citep{zhang2017quantum}.
However, both the final quantum states and the measurement output of a non-Markovian network are influenced by the integral process~\citep{ding2023quantumSIAM}, which is different from the traditional Markovian quantum systems. 
For example, when designing quantum measurement feedback control for a non-Markovian system, the measurement result can be influenced by the delayed phase due to a non-Markovian integral process~\citep{ding2023quantumSIAM}. For another example, the spectral property such as the wide linewidth of the input quantum noise in the QSDE can influence the integral kernel for the non-Markovian process and further the stochastic quantum dynamics~\citep{xue2016feedback}. The above quantum measurement and filtering approaches can be further generalized to other varied non-Markovian quantum networks~\citep{gough2012quantumPRAFiltering,jack1999non,gough2012single}.  

Among the physically realizable non-Markovian quantum filtering and feedback control approaches, a broad range of non-Markovian interactions manifest among different components in a quantum network with coherent feedback, which can be realized by the transmission and reflection of coherent fields, such as photons, in a waveguide with continuous modes or a cavity with infinite discrete modes~\citep{nemet2019comparison,ding2023quantum}. In the quantum coherent feedback realization based on waveguide quantum electrodynamics (waveguide-QED), the coherent fields emitted by one component in the quantum network can re-interact with the former emitter after being transmitted in the waveguide, resulting in linear quantum dynamical process with time delays~\citep{nemet2019comparison,ding2023quantum,DingNlevel2024,ding2023Automatica,ding2025transport}. 
Alternatively, considering that the continuous waveguide modes can be regarded as quantum noises with continuous frequency modes, the non-Markovian dynamics based on waveguide-QED can be interpreted from the perspective of quantum noises~\citep{gardiner2004quantum}. However, in such non-Markovian waveguide-QED networks with multiple time delays, the quantum noises cannot be simply regarded as \rm{It\={o}} noises as in the traditional Markovian circumstance~\citep{gardiner1985input,li2022control,fischer2018particle,yamamoto2012pure}, and this has not been systematically investigated.

In this paper, we study the non-Markovian quantum dynamics based on waveguide-QED from the perspective of quantum noises. In the multi-atom network coupled via a waveguide, the interactions between atoms and the waveguide can be equivalently modeled as the interactions between the atom network and multiple input quantum noise channels. The relationship among different quantum noise channels is affected by the non-Markovian properties of the quantum network such as spatial distributions and induced time delays, resulting in the non-Markovian commutative relationships different from traditional Markovian quantum systems. Then the control dynamics and filtering in the non-Markovian quantum network can be characterized by the non-Markovian properties of quantum noises. The rest of the paper is organized as follows. In 
Sec.~\ref{Sec:inputoutput}, we study the stochastic dynamics of the non-Markovian quantum networks with input quantum noises, based on Fig.~\ref{fig:NatomWaveguide}(a) where atoms are coupled to a semi-infinite waveguide at one point and Fig.~\ref{fig:NatomWaveguide}(b) that atoms are coupled to an infinite waveguide. The dynamics of these two proposals can be represented with non-Markovian integral processes with different integral kernels containing different time delays. Then we generalize to a more complex case where one atom can be coupled to a waveguide at several different points, as shown in Fig.~\ref{fig:NatomWaveguide}(c). 
In Sec.~\ref{Sec:NonMarknoise}, we study the \rm{It\={o}} properties of quantum noises in the above two waveguide-QED networks. In Sec.~\ref{Sec:QSDE}, we clarify how the quantum stochastic differential equations for the non-Markovian quantum network with time delays can be influenced by the integral kernels related to quantum noises, and their simplified formats within the Markovian approximation. Then in Sec.~\ref{Sec:filtering}, we further clarify the applications of the above results in the filtering of non-Markovian quantum networks. Sec.~\ref{Sec:conclusion} concludes this paper.

\section{Non-Markovian  dynamics realized by atoms coupled to a waveguide} \label{Sec:inputoutput}
\begin{figure}[htbp]
  \centering
  \centerline{\includegraphics[width=0.6\columnwidth]{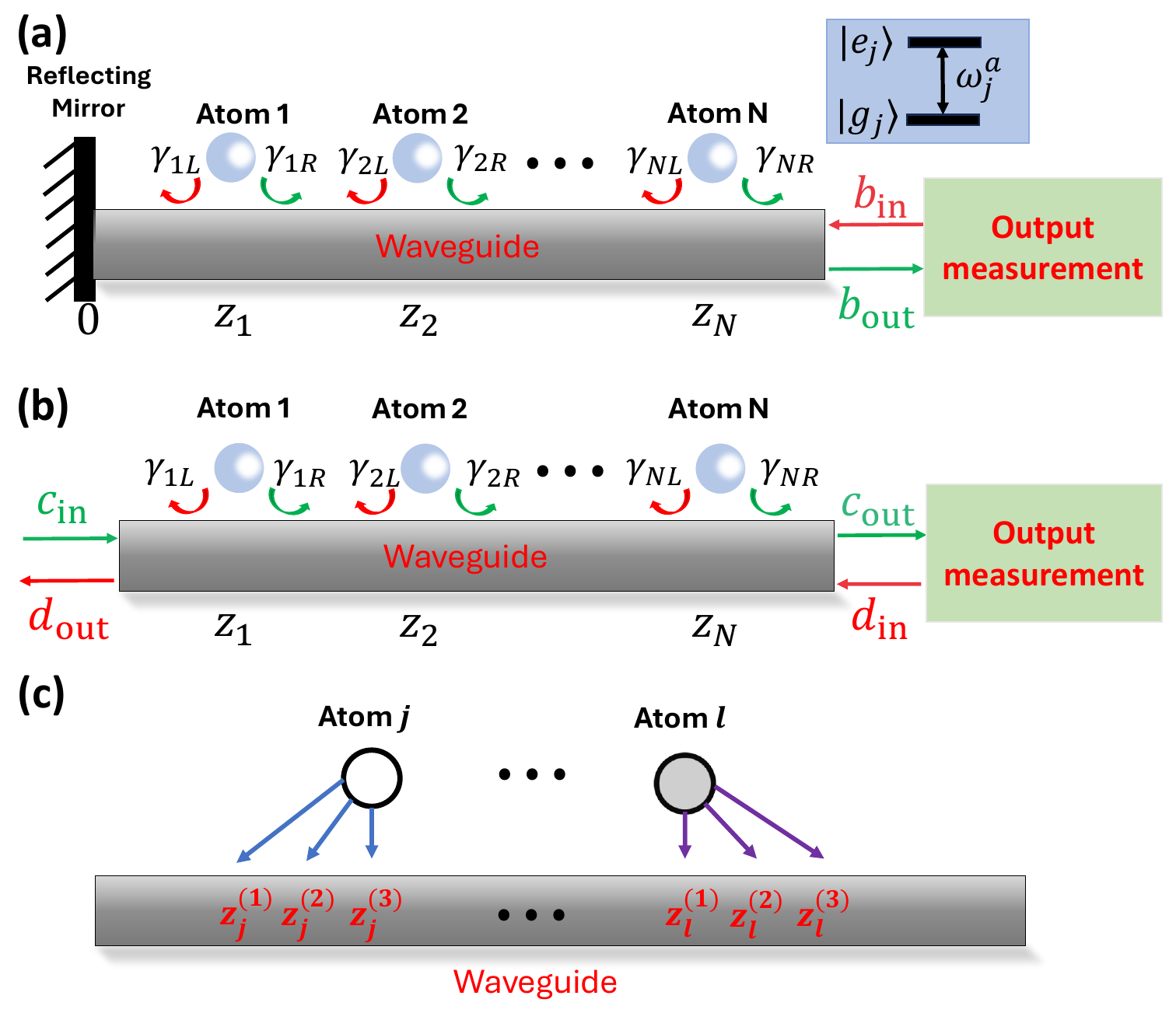}}
  \caption{Non-Markovian quantum network realizations: (a) Multiple two-level atoms coupled to a semi-infinite waveguide closed by the mirror. (b) Multiple two-level atoms coupled to an infinite waveguide. (c) Non-Markovian multi-point coupling between two-level systems and waveguide, where the waveguide can be semi-infinite as in (a) or infinite as in (b).}
  \label{fig:NatomWaveguide}
\end{figure}

The coherent and measurement feedback dynamics of the atom network can exhibit variations based on whether the atoms are coupled to a semi-infinite waveguide (as depicted in Fig.~\ref{fig:NatomWaveguide}(a)) or an infinite waveguide (as shown in Fig.~\ref{fig:NatomWaveguide}(b)). In scenario (a), the atoms' positions are evaluated by their distances to the mirror of a semi-infinite waveguide. Conversely, in scenario (b), where the atoms are coupled to an infinite waveguide, the quantum control dynamics are affected by the relative positions among atoms. For simplification, we represent the positions of atoms in Fig.~\ref{fig:NatomWaveguide}(b) to be the same as those in Fig.~\ref{fig:NatomWaveguide}(a), without affecting the subsequent results.

As illustrated in Fig.~\ref{fig:NatomWaveguide}(a), the $j$th atom can be coupled to a semi-infinite waveguide at the postion $z_j$, the input field, denoted as $b_{\rm in}$, is applied from the right end of the waveguide. Subsequently, the output field, $b_{\rm out}$, is obtained resulting from the interaction between the input field and atoms, as well as the reflection by the mirror at the left end.  However, when the atoms are coupled to an infinite waveguide, the right-propagating and left-propagating modes are independent of each other. This leads to the existence of two pairs of input-output fields, labeled by $c_{\rm in}$, $c_{\rm out}$ and $d_{\rm in}$, $d_{\rm out}$ in Fig.~\ref{fig:NatomWaveguide}(b), respectively. Besides, the atoms in Fig.~\ref{fig:NatomWaveguide}(a) and Fig.~\ref{fig:NatomWaveguide}(b) can also be coupled to a waveguide at multiple points, as illustrated by Fig.~\ref{fig:NatomWaveguide}(c). 

Here we take the $j$th two-level atom in Fig.~\ref{fig:NatomWaveguide}(a) as an example to provide a brief overview of atomic dynamics. The state $\left|e_j\right\rangle$ denotes the excited state of the $j$th atom, and $\left|g_j\right\rangle$ represents its ground state. The raising operator $\sigma_j^+ = \left|e_j\right\rangle \left\langle g_j \right|$ signifies the transition of the atom from the ground state to the excited state. Conversely, the lowering operator $\sigma_j^- = \left|g_j\right\rangle \left\langle e_j \right|$ indicates the reverse process where the atom can decay to the ground state from its excited state. $\sigma_j^+ \sigma_j^-$ represents an excited two-level atom, and the Pauli Z operator $\sigma_j^z = \sigma_j^+ \sigma_j^- - \sigma_j^- \sigma_j^+$ quantifies the difference in populations between the excited and ground states of the $j$th atom. 
In the following, $[A,B] = AB-BA$ represents the commutator between two operators $A$ and $B$, which can be an atomic operator or a quantum noise operator in the waveguide.

\subsection{Atoms coupled to a semi-infinite waveguide} \label{Sec:InputOutputSemi}
Assume that the quantum system under consideration (i.e., the atom network in Fig.~\ref{fig:NatomWaveguide}) exists in the Hilbert space denoted as $\mathcal{H}_S$, and the quantum noise field (i.e., the photon in the waveguide) occupies the Hilbert space $\mathcal{H}_F$, they can construct a joint system expressed in a tensor product format $\mathcal{H}_S\otimes \mathcal{H}_F$~\citep{dong2019response,fischer2018particle}. Consequently, the quantum state of the joint system comprising the atom network and a semi-infinite waveguide in Fig.~\ref{fig:NatomWaveguide}(a) can be represented as $|\Psi\rangle= |\psi,\mathbf{k}\rangle  = |\psi\rangle \otimes  |\mathbf{k}\rangle $ where $|\psi\rangle$ signifies the atom's state and $|\mathbf{k}\rangle$ represents the state of the waveguide.

Consider the case with $N$ two-level atoms in Fig.~\ref{fig:NatomWaveguide}(a), the Hamiltonian of the system reads~\citep{zheng2013persistent}  
\begin{equation} \label{con:Hspicture}
\begin{aligned}
H/\hbar &=\sum_{j=1}^N \left(\omega_j^a - i\frac{\eta_j}{2}\right) \sigma_j^+\sigma_j^- +H_{\omega}/\hbar
 +  \sum_{j=1}^N H_{\rm I}^{(j)}/\hbar ,
\end{aligned}
\end{equation}
where the first part corresponds to the Hamiltonian of multiple two-level atoms with the resonant frequency $\omega_j^a$, and the decay rate of the $j$th atom to the environment is $\eta_j$. The second part, $H_{\omega}/\hbar
=\int_{-\infty}^{\infty}  \omega  d_{\omega}^{\dag}d_{\omega}\mathrm{d}\omega$, represents the Hamiltonian of a semi-infinite waveguide consisting of a continuum modes, where $d_{\omega}^{\dag} \left(d_{\omega}\right)$ denotes the creation (annihilation) operator of the waveguide mode $\omega$. The commutator between two distinct waveguide modes satisfies $\left[d_{\omega} (\omega),d_{\omega}^{\dag} (\omega') \right] = \delta(\omega-\omega')$~\citep{arranz2021modeling}. The last part describes the interactions between the atom network and the waveguide. For example, the interaction Hamiltonian between the $j$th atom and waveguide can be expressed as~\citep{ZhangBin} 
\begin{equation} \label{con:Hintjchiral}
\begin{aligned}
H_{\rm I}^{(j)} /\hbar 
&= \int_{-\infty}^{\infty}  \left [g_{j}(\omega,z_j) d^{\dag}_{\omega}\sigma_j^-  + g_{j}^*(\omega,z_j) d_{\omega}\sigma_j^+\right ]\mathrm{d}\omega,
\end{aligned}
\end{equation}
and
\begin{equation} \label{con:gj}
\begin{aligned}
g_{j}(\omega,z_j)&= i\left (\sqrt{\frac{\gamma_{jR}}{2\pi}} e^{-i\omega z_j/c}-\sqrt{\frac{\gamma_{jL}}{2\pi}} e^{i\omega z_j/c}\right),
\end{aligned}
\end{equation}
where $\gamma_{jR}$ is the coupling strength between the $j$th atom and the right-propagating mode in the waveguide, while $\gamma_{jL}$ denotes the coupling strength between the $j$th atom and the left-propagating mode in the waveguide. The Planck constant $\hbar$ is set to be 1 in this paper and is omitted in the following.

The dynamics of an arbitrary operator $O$ in the quantum system is determined by the Heisenberg equation $\dot{O} = -i [O,H]$. Subsequently, we can obtain the equations for the atomic operators as
\begin{small}
\begin{equation} \label{con:HeisenSm}
\begin{aligned}
\dot{\sigma}_j^-=& -\left(i\omega_j^a + \frac{\eta_j}{2} \right) \sigma_j^-  +\int_{-\infty}^{\infty}   \left (\sqrt{\frac{\gamma_{jR}}{2\pi}}  e^{i\omega z_j/c} -\sqrt{\frac{\gamma_{jL}}{2\pi}}  e^{-i\omega z_j/c}\right) d_{\omega}\sigma_j^z \mathrm{d}\omega,
\end{aligned}
\end{equation}
\end{small}%
and the equation for the photon mode $\omega$ in the waveguide is given by
\begin{equation} \label{con:HeisendOmega}
\begin{aligned}
\dot{d}_{\omega}(\omega,t)& = -i\omega d_{\omega}(\omega,t) +i \sum_j g_{j}(\omega,z_j) \sigma_j^-.
\end{aligned}
\end{equation}
Thus, $d_{\omega}(\omega,t)$ can be determined by integrating Eq.~(\ref{con:HeisendOmega}) as
\begin{equation} \label{con:domegatSolution}
\begin{aligned}
d_{\omega}(\omega,t) =& e^{-i\omega\left(t-t_0\right)}d_{\omega}(\omega,t_0)  + i \sum_j g_{j}(\omega,z_j)\int_{t_0}^t e^{-i\omega \left(t-t' \right)} \sigma_j^-\left (t' \right)\mathrm{d}t',
\end{aligned}
\end{equation}
where $t_0$ represents the initial time of the evolution.

As stated by~\citep{fan2010input}, the input field in the waveguide can be derived in the time domain according to the initial value of the operator $d_{\omega}(\omega,t_0)$ as
\begin{subequations} \label{eq:FreqOperator}
\begin{numcases}{}
b_{\rm in}(t) = \frac{1}{\sqrt{2\pi}} \int_{-\infty}^{\infty}  d_{\omega}(\omega,t_0) e^{-i\omega \left(t-t_0\right)}\mathrm{d}\omega,\\
d_{\omega}(\omega,t_0) = \frac{1}{\sqrt{2\pi}} \int_{t_0}^{\infty} b_{\rm in}(\nu) e^{i\omega \left(\nu-t_0\right)}\mathrm{d}\nu,
\end{numcases}
\end{subequations}
where $\left[b_{\rm in}(t),b_{\rm in}^{\dag}(t') \right] = \delta(t-t')$~\citep{gardiner1985input}.

Similarly, by integrating Eq.~(\ref{con:HeisendOmega}) up to a final time $t_1>t$~\citep{fan2010input,zhang2013non}, the output field in the waveguide $b_{\rm out}(t)$ can be defined as the Fourier transform of $d_{\omega}\left(\omega,t_1\right)$. Consequently, the output field $b_{\rm out}(t)$ can be expressed as~\citep{ding2025transport}
\begin{equation} \label{con:boutt}
\begin{aligned}
b_{\rm out}(t) &= \frac{1}{\sqrt{2\pi}} \int_{-\infty}^{\infty}  d_{\omega}(\omega,t_1) e^{-i\omega \left(t-t_1\right)}\mathrm{d}\omega  \\
&= \frac{1}{\sqrt{2\pi}} \int_{-\infty}^{\infty}  \Big[  e^{-i\omega\left(t_1-t_0\right)}d_{\omega}(\omega,t_0)  +  i \sum_j g_{j}(\omega,z_j)\int_{t_0}^{t_1} e^{-i\omega \left(t_1-t' \right)} \sigma_j^-\left (t' \right)\mathrm{d}t'\Big] e^{-i\omega \left(t-t_1\right)}\mathrm{d}\omega  \\
&=b_{\rm in}(t)  + \sum_j     \left [\sqrt{\gamma_{jL}} \sigma_j^-\left (t-\frac{z_j}{c} \right)-\sqrt{\gamma_{jR}} \sigma_j^-\left (t+\frac{z_j}{c}\right)\right],
\end{aligned}
\end{equation}
where the integration
\begin{small}
\begin{equation} \label{con:integralCal}
\begin{aligned}
&\frac{i}{\sqrt{2\pi}} \int_{-\infty}^{\infty}    g_{j}(\omega,z_j)\int_{t_0}^{t_1} e^{-i\omega \left(t-t' \right)} \sigma_j^-\left (t' \right)  \mathrm{d}t'\mathrm{d}\omega =\int_{t_0}^{t_1}     \left [\sqrt{\gamma_{jL}} \delta\left(t-t'  -\frac{z_j}{c} \right)-\sqrt{\gamma_{jR}} \delta\left(t-t' + \frac{z_j}{c} \right)\right]\sigma_j^-\left (t' \right) \mathrm{d}t'.
\end{aligned}
\end{equation}
\end{small}%
Eq.~(\ref{con:integralCal}) can be regarded as an integral process involving two kernels in the format of delta functions. Similar to the approach in~\citep{domokos2002PRA}, the term $\sigma_j^-\left (t-z_j/c \right)$ corresponds to the forward propagating modes, while the term $\sigma_j^-\left (t+z_j/c\right)$ corresponds to the backward propagating modes.

\subsection{Atoms coupled to an infinite waveguide}\label{Sec:InputOutputInfiW}
In Fig.~\ref{fig:NatomWaveguide}(b), the joint system consisting of the atom network and an infinite waveguide can be represented as $\left|\tilde{\Psi}\right \rangle= \left|\tilde{\psi},\mathbf{R},\mathbf{L}\right \rangle  = \left|\tilde{\psi} \right \rangle \otimes  |\mathbf{R}\rangle  \otimes  |\mathbf{L}\rangle$, where $\left|\tilde{\psi} \right \rangle$ represents the atomic state, $|\mathbf{R}\rangle$ denotes the right-propagating photon state in the waveguide, and $|\mathbf{L}\rangle$ signifies the left-propagating photon state in the waveguide. 
When $N$ two-level atoms are coupled to an infinite waveguide as shown in Fig.~\ref{fig:NatomWaveguide}(b), the Hamiltonian is given by~\citep{lalumiere2013input,pichler2016photonic}
\begin{equation} \label{con:freeHam0}
\begin{aligned}
\tilde{H} =&\sum_{j=1}^N \left(\omega_j^a - i\frac{\eta_j}{2}\right)\sigma_j^+\sigma_j^- + \tilde{H}_{\omega}
  +  \sum_{j=1}^N \tilde{H}_{\rm I}^{(j)},
\end{aligned}
\end{equation}
where the first component comprises the atomic Hamiltonian align with Eq.~(\ref{con:Hspicture}). The second component, $ \tilde{H}_{\omega}
  = \int_{0}^{\infty}  \omega l_{\omega}^{\dag}l_{\omega}\mathrm{d}\omega + \int_{0}^{\infty}  \omega r_{\omega}^{\dag}r_{\omega} \mathrm{d}\omega$, characterizes the Hamiltonian of an infinite waveguide encompassing left-propagating and right-propagating modes. Here, $l_{\omega}^{\dag} \left(l_{\omega}\right)$ signifies the creation (annihilation) operator of the left-propagating waveguide mode $\omega$, while $r_{\omega}^{\dag} \left(r_{\omega}\right)$ denotes that for the right-propagating waveguide mode $\omega$, and $\left[l_{\omega} (\omega),l_{\omega}^{\dag} (\omega') \right] =\left[r_{\omega} (\omega),r_{\omega}^{\dag} (\omega') \right]  =\delta(\omega-\omega')$. Lastly, the final component in Eq.~(\ref{con:freeHam0}) encapsulates the interaction between the atom and the left- or right-propagating waveguide modes as~\citep{lalumiere2013input,pichler2016photonic}
\begin{equation} \label{con:Hintj2nonchiral}
\begin{aligned}
\tilde{H}_{\rm I}^{(j)} &=\int_{-\infty}^{\infty}  \left [\bar{g}_{j}(\omega,z_j) l^{\dag}_{\omega}\sigma_j^-  + \bar{g}_{j}^*(\omega,z_j) l_{\omega}\sigma_j^+\right ] \mathrm{d}\omega  + \int_{-\infty}^{\infty} \left [\hat{g}_{j}(\omega,z_j) r^{\dag}_{\omega}\sigma_j^-   + \hat{g}_{j}^*(\omega,z_j) r_{\omega}\sigma_j^+ \right ]\mathrm{d}\omega,
\end{aligned}
\end{equation}
where the coupling strength between the atom and waveguide depends on the atom positions as $\bar{g}_{j}(\omega,z_j)= -i\sqrt{\gamma_{jL}/2\pi} e^{i\omega z_j/c}$ and $\hat{g}_{j}(\omega,z)= -i\sqrt{\gamma_{jR}/2\pi} e^{-i\omega z_j/c}$, representing the coupling between an atom at $z_j$ and the left- or right-propagating waveguide modes respectively, which is similar to those in Eq.~(\ref{con:gj}).

In analog to Eq.~(\ref{con:HeisenSm}), we can derive the Heisenberg equations for the atomic operators as
\begin{equation} \label{con:HeisenSmInfinite}
\begin{aligned}
\dot{\sigma}_j^-
=& -i\omega_j^a \sigma_j^-  - \frac{\eta_j}{2} \sigma_j^- -\int_{-\infty}^{\infty}  \left (\sqrt{\frac{\gamma_{jR}}{2\pi}} e^{i\omega z_j/c} r_{\omega} + \sqrt{\frac{\gamma_{jL}}{2\pi}}  e^{-i\omega z_j/c}l_{\omega}\right) \sigma_j^z \mathrm{d}\omega,
\end{aligned}
\end{equation}
and the evolution of the operators $l_{\omega}$ and $r_{\omega}$ reads
\begin{subequations} \label{con:lomegacomega}
\begin{numcases}{}
\dot{l}_{\omega}(\omega,t) = -i\omega l_{\omega}(\omega,t) +i \sum_j \bar{g}_{j}(\omega,z_j) \sigma_j^-, \label{lOmega}\\
\dot{r}_{\omega}(\omega,t) = -i\omega r_{\omega}(\omega,t) +i \sum_j \hat{g}_{j}(\omega,z_j) \sigma_j^-.\label{cOmega}
\end{numcases}
\end{subequations}

Then the two input fields $d_{\rm in}(t)$ and $c_{\rm in}(t)$ in Fig.~\ref{fig:NatomWaveguide}(b) can be defined as
\begin{subequations} \label{eq:InfWInputFields}
\begin{numcases}{}
d_{\rm in}(t) = \frac{1}{\sqrt{2\pi}} \int_{-\infty}^{\infty}  l_{\omega}(\omega,t_0) e^{-i\omega \left(t-t_0\right)}\mathrm{d}\omega, \label{dint}\\
l_{\omega}(\omega,t_0) = \frac{1}{\sqrt{2\pi}} \int_{t_0}^{\infty} d_{\rm in}(\nu) e^{i\omega \left(\nu-t_0\right)}\mathrm{d}\nu, \label{lomega}\\  
c_{\rm in}(t) = \frac{1}{\sqrt{2\pi}}\int_{-\infty}^{\infty}  r_{\omega}(\omega,t_0) e^{-i\omega \left(t-t_0\right)}\mathrm{d}\omega, \label{cint}\\
r_{\omega}(\omega,t_0) = \frac{1}{\sqrt{2\pi}} \int_{t_0}^{\infty} c_{\rm in}(\nu) e^{i\omega \left(\nu-t_0\right)}\mathrm{d}\nu.\label{cOmega2}
\end{numcases}
\end{subequations}
Similar to Eq.~(\ref{con:boutt}), the left-propagating output field can be represented as
\begin{equation} \label{con:bouttInFWdout}
\begin{aligned}
&d_{\rm out}(t) =d_{\rm in}(t)  +\sum_j   \sqrt{\gamma_{jL}}\sigma_j^- \left(t- z_j/c \right),
\end{aligned}
\end{equation}
and the right-propagating output field reads
\begin{equation} \label{con:bouttInFWcout}
\begin{aligned}
c_{\rm out}(t) &=c_{\rm in}(t)  +\sum_j    \sqrt{\gamma_{jR}}\sigma_j^- \left(t+ z_j/c \right).
\end{aligned}
\end{equation}

\subsection{Multi-point coupling between atoms and waveguide}
When the $j$th atom is coupled to the waveguide at multiple points $z_j^{(1)}, z_j^{(2)},\cdots,z_j^{(\mathbf{n}_j)}$ with $z_j^{(1)} < z_j^{(2)} < \cdots <  z_j^{(\mathbf{n}_j)}$  as in ~\citep{WangZH2020PRA,DuliPRAMultiPointCoup} or Fig.~\ref{fig:NatomWaveguide}(c), the free Hamiltonian for the semi-infinite and infinite waveguide circumstance are the same as that in Eq.~(\ref{con:Hspicture}) and Eq.~(\ref{con:freeHam0}), respectively. However, the interaction Hamiltonian will be different. Specifically, when $\mathbf{n}_j = 2$ for an arbitrary $j$, the coupling method reduces to the case in ~\citep{kockum2018decoherence}.

Assume that the $j$th atom can be coupled to the right-propagating and left-propagating modes in the waveguide via its $n$th connecting point at $z_j^{(n)}$ with the strengths $\gamma_{jR}^{(n)}$ and $\gamma_{jL}^{(n)}$ respectively. Then the input-output relationships for the semi-infinite and infinite waveguide are as follows.

\subsubsection{Multi-point coupling based on semi-infinite waveguide}
The interaction Hamiltonian for the $j$th atom coupled to the waveguide at $z_j^{(n)}$ can be generalized from Eq.~(\ref{con:Hintjchiral}) as
\begin{equation} \label{con:HintjchiralGiant}
\begin{aligned}
H_{\rm m,I}^{(j)} 
&= \int_{-\infty}^{\infty} \sum_{n=1}^{\mathbf{n}_j} \left ( g_{j}^{(n)}  d^{\dag}_{\omega}\sigma_j^-  +  {g_{j}^{(n)}}^* d_{\omega}\sigma_j^+\right )\mathrm{d}\omega,
\end{aligned}
\end{equation}
with the subscript $\rm m,I$ represents interaction Hamiltonian via multi-point coupling, and 
\begin{equation} \label{con:gjn}
\begin{aligned}
g_{j}^{(n)}  =  i\left [\sqrt{\frac{\gamma_{jR}^{(n)}}{2\pi}} e^{-i\omega z_j^{(n)}/c}-\sqrt{\frac{\gamma_{jL}^{(n)}}{2\pi}} e^{i\omega z_j^{(n)}/c}\right],
\end{aligned}
\end{equation}
generalized from the definition of $ g_{j}(\omega,z_j)$ in Eq.~(\ref{con:gj}). In addition, when $\mathbf{n}_j \equiv 1$, the above coupling method reduces to the case in Sec.~\ref{Sec:InputOutputSemi}. Due to Sec.~\ref{Sec:InputOutputSemi}, we can similarly derive the non-Markovian input-output relationship based on Eq.~(\ref{con:boutt}) as
\begin{equation} \label{con:bouttMultiPoint}
\begin{aligned}
&b_{\rm m,out}(t) =b_{\rm m,in}(t)  +\sum_j \sum_{n=1}^{\mathbf{n}_j}    \left [\sqrt{\gamma_{jL}^{(n)}} \sigma_j^-\left (t- z_j^{(n)}/c \right)-\sqrt{\gamma_{jR}^{(n)}} \sigma_j^-\left (t+z_j^{(n)}/c\right)\right],
\end{aligned}
\end{equation}
where $b_{\rm m,in}(t)$ and $b_{\rm m,out}(t)$ represent the input and output fields, respectively.

\subsubsection{Multi-point coupling based on infinite waveguide} Similarly, the multi-point interactions between atoms and an infinite waveguide can be represented with the following Hamiltonian 
\begin{small}
\begin{equation} \label{con:HintjchiralGiant2}
\begin{aligned}
\tilde{H}_{\rm m,I}^{(j)}  = & \sum_{n=1}^{\mathbf{n}_j} \int_{-\infty}^{\infty}  \left [\bar{g}_{j}^{(n)} l^{\dag}_{\omega}\sigma_j^- +  \left(\bar{g}_{j}^{(n)}\right)^* l_{\omega}\sigma_j^+ \right ] \mathrm{d}\omega +  \sum_{n=1}^{\mathbf{n}_j} \int_{-\infty}^{\infty}  \left [\hat{g}_{j}^{(n)} r^{\dag}_{\omega}\sigma_j^-   + \left(\hat{g}_{j}^{(n)}\right)^* r_{\omega}\sigma_j^+ \right ]\mathrm{d}\omega,
\end{aligned}
\end{equation}
\end{small}%
where the coupling strengths between the atom and waveguide are $\bar{g}_{j}^{(n)} = -i\sqrt{\gamma_{jL}^{(n)}/2\pi} e^{i\omega z_j^{(n)}/c}$ and $\hat{g}_{j}^{(n)}= -i\sqrt{\gamma_{jR}^{(n)}/2\pi} e^{-i\omega z_j^{(n)}/c}$. Generalized from Eqs.~(\ref{con:bouttInFWdout},\ref{con:bouttInFWcout}), we can derive the input-output relationship as
\begin{subequations} \label{eq:InputOutputMultiPInfintie}
\begin{numcases}{}
d_{\rm m,out}(t) = d_{\rm m,in}(t) + \sum_j \sum_{n=1}^{\mathbf{n}_j}   \sqrt{\gamma_{jL}^{(n)}}\sigma_j^- \left(t- z_j^{(n)}/c \right),  \\
c_{\rm m,out}(t)= c_{\rm m,in}(t)+ \sum_j\sum_{n=1}^{\mathbf{n}_j} \sqrt{\gamma_{jR}^{(n)}}\sigma_j^- \left(t+ z_j^{(n)}/c \right),  
\end{numcases}
\end{subequations}
where $d_{\rm m,in}(t)$ ($c_{\rm m,in}(t)$) and $d_{\rm m,out}(t)$ ($c_{\rm m,out}(t)$) represent the input and output fields for left-propagating (right-propagating) fields, respectively.

\subsection{Non-Markovian memory kernel for atom-waveguide interactions}
Using the above input-output relationship based on atom-waveguide interactions, we can derive the following memory kernels for the above non-Markovian interactions. 
\subsubsection{Memory kernel for semi-infinite waveguide case}
When the atoms are coupled to a semi-infinite waveguide as in Fig.~\ref{fig:NatomWaveguide}(a), we can derive the atomic dynamics by defining atoms' decay to the waveguide with the time-varying operator $\mathbf{L}(t)$ according to the input-output relationship in Eq.~(\ref{con:boutt}) as
\begin{equation} \label{con:LSemiDef}
\begin{aligned}
\mathbf{L}(t) = & \sum_j     \left [\sqrt{\gamma_{jL}} \sigma_j^-\left (t-\frac{z_j}{c} \right)-\sqrt{\gamma_{jR}} \sigma_j^-\left (t+\frac{z_j}{c}\right)\right]\\
 \triangleq & \sum_j \int_0^t \kappa_j (t-\nu)  \sigma_j^-\left (\nu\right) \mathrm{d}\nu,
\end{aligned}
\end{equation}
where the integral kernel
\begin{equation} \label{con:kappajDef}
\begin{aligned}
\kappa_j (t-\nu)  = \left[ \sqrt{\gamma_{jL}} \delta\left(t- \nu-\frac{z_j}{c}  \right) - \sqrt{\gamma_{jR}} \delta \left(t- \nu+\frac{z_j}{c}  \right) \right].
\end{aligned}
\end{equation}
Then Eq.~(\ref{con:boutt}) can be equivalently rewritten in a general format as
\begin{equation} \label{con:bouttV2}
\begin{aligned}
b_{\rm out}(t) &= b_{\rm in}(t) +\mathbf{L}(t)\\
&=b_{\rm in}(t) + \sum_j \int_{t_0}^t \kappa_j (t-\nu)  L_j\left (\nu\right) \mathrm{d}\nu,
\end{aligned}
\end{equation}
where $L_j= \sigma_j^-$ represents the lowering operator for the $j$th atom, and this generalizes the format in \citep{zhang2013non,diosi2012non} to the distributed network circumstance with multiple memory kernels. Besides, when atoms are coupled to the semi-infinite waveguide via multi-point coupling, the integral kernel reads
\begin{equation} \label{con:kappajDefMlultiPoint}
\begin{aligned}
\mathbf{\kappa}_j (t-\nu)  
= \sum_{n=1}^{\mathbf{n}_j} \left[  \sqrt{\gamma_{jL}^{(n)}} \delta\left(t- \nu-\frac{z_j^{(n)}}{c}  \right) - \sqrt{\gamma_{jR}^{(n)}} \delta \left(t- \nu+\frac{z_j^{(n)}}{c}  \right) \right],
\end{aligned}
\end{equation}
and the interacting operator $L_j$ is the same as that in Eq.~(\ref{con:bouttV2}).

By concluding Eqs.~(\ref{con:bouttV2},\ref{con:kappajDefMlultiPoint}), the non-Markovian interactions between a waveguide and atoms or atom like objects are only influenced by the coupling operators $L_j$ and the integral kernels $\mathbf{\kappa}_j$. This can be extended to encompass non-Markovian couplings between waveguide and various systems such as cavities~\citep{duda2024efficient,ding2023quantum,DingNlevel2024,ding2025transport} or magnons~\citep{zhan2022chirality}, merely by substituting the operators $L_j$ and integral kernels.

\subsubsection{Memory kernel for infinite waveguide case}
When the atoms are coupled to an infinite waveguide, as shown in Fig.~\ref{fig:NatomWaveguide}(b), the decay of atoms towards the waveguide can be divided into the components propagating to the left and to the right. Based on Eqs.~(\ref{con:bouttInFWdout},\ref{con:bouttInFWcout}), we define the following two operators as
\begin{subequations} \label{con:LeftRightDecayOperator}
\begin{numcases}{}
\mathbf{l} (t) \triangleq \sum_j \int_{t_0}^t  \sigma_j^- (\nu)  \kappa_j^{l}(t-\nu)\mathrm{d}\nu, \label{leftdecay}\\
\mathbf{r} (t) \triangleq \sum_j \int_{t_0}^t \sigma_j^- (\nu)\kappa_j^{r}(t-\nu) \mathrm{d}\nu,\label{rightdecay2}
\end{numcases}
\end{subequations}
where the integral kernel reads $\kappa_j^{l}(t-\nu) = \sqrt{\gamma_{jL}} \delta\left( t- \nu-z_j/c  \right) $ and $\kappa_j^{r}(t-\nu) = \sqrt{\gamma_{jR}} \delta\left( t- \nu+ z_j/c  \right) $. Then the input-output equations~(\ref{con:bouttInFWdout}) and (\ref{con:bouttInFWcout}) can be rewritten as
\begin{subequations} \label{con:InfiniteOutputKernel}
\begin{numcases}{}
d_{\rm out}(t) = d_{\rm in}(t) +\sum_j \int_{t_0}^t  L_j (\nu)  \kappa_j^{l}(t-\nu)\mathrm{d}\nu, \label{leftdecay2}\\
c_{\rm out}(t) = c_{\rm in}(t)  +  \sum_j \int_{t_0}^t L_j (\nu)\kappa_j^{r}(t-\nu) \mathrm{d}\nu,\label{rightdecay}
\end{numcases}
\end{subequations}
where $L_j$ has been defined in Eq.~(\ref{con:bouttV2}).
When generalizing to the case where atoms are coupled to an infinite waveguide via multiple points, we only need to replace the integral kernels with
\begin{small}
\begin{subequations} \label{con:kappajDefMlultiPointInfinite}
\begin{numcases}{}
\mathbf{\kappa}_j^l (t-\nu)  = \sum_{n=1}^{\mathbf{n}_j} \sqrt{\gamma_{jL}^{(n)}} \delta\left(t- \nu-\frac{z_j^{(n)}}{c}  \right),\\
\mathbf{\kappa}_j^r (t-\nu)  = \sum_{n=1}^{\mathbf{n}_j} \sqrt{\gamma_{jR}^{(n)}} \delta\left(t- \nu+\frac{z_j^{(n)}}{c}  \right).
\end{numcases}
\end{subequations}
\end{small}%

\section{Non-Markovian quantum noise properties based on atom-waveguide interactions} \label{Sec:NonMarknoise}
Based on the input-output relationships in Sec.~\ref{Sec:inputoutput}, this section examines the characteristics of quantum noises in the waveguide arising from the non-Markovian interactions between an atom network and a semi-infinite or infinite waveguide.

\subsection{Non-Markovianity of quantum noises in a semi-infinite waveguide}

To elucidate the non-Markovian characteristic concerning quantum noises, the Hamiltonian in Eq.~(\ref{con:Hspicture}) can be equivalently rewritten by tracing out the waveguide Hamiltonian, as in Appendix~\ref{Sec:AppdixTraceWaveguide}.  Subsequently, the interactions between the atoms and the waveguide can be perceived as interactions between the atoms and time-varying input noises as ~\citep{gardiner1985input,fischer2018particle}
\begin{equation} \label{con:HRWProduct}
\begin{aligned}
\mathbf{H}(t)
=& \sum_{j=1}^N \left(\omega_j^a - i\frac{\eta_j}{2}\right) \sigma_j^+\sigma_j^-  \\
&+ i\sum_{j=1}^N \left[ \sqrt{\gamma_{jR}} \sigma_j^-   b_{\rm in}^{\dag}\left(t-z_j/c\right)   -\sqrt{\gamma_{jL}}   \sigma_j^-  b_{\rm in}^{\dag}\left(t+ z_j/c\right)\right] \\
&- i\sum_{j=1}^N \left[ \sqrt{\gamma_{jR}} \sigma_j^+   b_{\rm in}\left(t-z_j/c\right)   -  \sqrt{\gamma_{jL}}   \sigma_j^+  b_{\rm in}\left(t+ z_j/c\right)\right],
\end{aligned}
\end{equation}
where $b_{\rm in}\left(t-z_j/c\right)$ and  $b_{\rm in}\left(t+z_j/c\right)$ are delayed noise operators for the waveguide modes by regarding the right-propagating direction as positive based on Eq.~(\ref{eq:FreqOperator}), and more details are given in Appendix~\ref{Sec:AppdixTraceWaveguide}. Further in Eq.~(\ref{con:HRWProduct}), we denote
\begin{equation} \label{con:Hsdef}
\begin{aligned}
H_s = \sum_{j=1}^N \left(\omega_j^a - i\frac{\eta_j}{2}\right) \sigma_j^+\sigma_j^-,
\end{aligned}
\end{equation}
\begin{equation} \label{con:Multichannle}
\begin{aligned}
{b_{\rm in}^{(j)}}^{\dag} (t)=&\sqrt{\gamma_{jR}} b_{\rm in}^{\dag}\left(t-\frac{z_j}{c}\right)   - \sqrt{\gamma_{jL}} b_{\rm in}^{\dag}\left(t+ \frac{z_j}{c}\right)\\
 \triangleq &  \int_0^t \tilde{\kappa}_j (t-\nu)  b_{\rm in}^{\dag}\left (\nu\right)\mathrm{d}\nu,
\end{aligned}
\end{equation}
and
\begin{equation} \label{con:kappajDef2}
\begin{aligned}
\tilde{\kappa}_j (t-\nu)  = \left[ \sqrt{\gamma_{jR}} \delta\left(t- \nu-\frac{z_j}{c}  \right) -\sqrt{\gamma_{jL}} \delta \left(t- \nu+\frac{z_j}{c}  \right) \right].
\end{aligned}
\end{equation}

Then the Hamiltonian in Eq.~(\ref{con:HRWProduct}) can be rewritten as
\begin{equation} \label{con:HeffSemiInputFieldKappa}
\begin{aligned}
\mathbf{H}(t)
& = H_s +  i\sum_{j=1}^N \left\{ \left[\int_0^t \tilde{\kappa}_j (t-\nu)  b_{\rm in}^{\dag}\left (\nu\right)\mathrm{d}\nu\right]  L_j -\rm H.c.\right\},
\end{aligned}
\end{equation}
where $\rm H.c.$ represents Hermitian conjugation.

The components with integral processes in Eq.~(\ref{con:HeffSemiInputFieldKappa}) can be interpreted as the interactions between the atom network and $N$ input quantum noise channels, denoted by the operators $b_{\rm in}^{(j)}(t)$ and ${b_{\rm in}^{(j)}}^{\dag} (t)$. According to Eq.~(\ref{con:Multichannle}), the commutator between arbitrary two quantum noise channels satisfies that
\begin{equation} \label{con:bjrelation}
\begin{aligned}
 \left[ b_{\rm in}^{(j)} (t), {b_{\rm in}^{(l)}}^{\dag} (t') \right] =&\int_0^t\int_0^{t'} \tilde{\kappa}_j^* (t-\nu) \tilde{\kappa}_l (t'-\nu') \left[ b_{\rm in}\left (\nu\right) b_{\rm in}^{\dag}\left (\nu'\right) -b_{\rm in}^{\dag}\left (\nu'\right)  b_{\rm in}\left (\nu\right) \right]\mathrm{d}\nu'\mathrm{d}\nu\\
=&\int_0^t \tilde{\kappa}_j^* (t-\nu) \tilde{\kappa}_l (t'-\nu) \mathrm{d}\nu.
\end{aligned}
\end{equation}

Then we have the following lemmas for non-Markovian commutative relationships among different quantum noise channels based on interactions between atoms and a semi-infinite waveguide, which is different from the Markovian circumstance in~\citep{gardiner1985input}.

\newtheorem{lemma}{Lemma}

\begin{lemma} \label{lemmabinjt}
For the non-Markovian Hamiltonian in Eq.~(\ref{con:HeffSemiInputFieldKappa}), consider the $j$th input quantum noise channel, the commutator for the quantum noise in the time domain reads
\begin{equation} 
\begin{aligned} \label{con:Eqlemma1}
&\left[ b_{\rm in}^{(j)} (t) ,{b_{\rm in}^{(j)}}^{\dag} (t')\right] =  \left(\gamma_{jR} + \gamma_{jL} \right) \delta\left(t-t' \right) -\sqrt{\gamma_{jL}\gamma_{jR}} \delta\left( t - t' +\frac{2z_j}{c}\right) -\sqrt{\gamma_{jL}\gamma_{jR}} \delta\left( t - t' -\frac{2z_j}{c}\right) .
\end{aligned}
\end{equation}
\end{lemma}
\begin{proof}
When $j=l$ in Eq.~(\ref{con:bjrelation}), 
\begin{equation} \label{con:ttpexampleProof}
\begin{aligned}
&\int_0^t \tilde{\kappa}_j^* (t-\nu) \tilde{\kappa}_j (t'-\nu) \mathrm{d}\nu \\
=&\int_0^t \left[ \sqrt{\gamma_{jR}} \delta\left(t- \nu-\frac{z_j}{c}  \right) -\sqrt{\gamma_{jL}} \delta \left(t- \nu+\frac{z_j}{c}  \right) \right] \left[ \sqrt{\gamma_{jR}} \delta\left(t'- \nu-\frac{z_j}{c}  \right) -\sqrt{\gamma_{jL}} \delta \left(t'- \nu+\frac{z_j}{c}  \right) \right]\mathrm{d}\nu\\
=&\left(\gamma_{jR} + \gamma_{jL} \right) \delta(t-t')  -\sqrt{\gamma_{jR}\gamma_{jL}}\delta\left(t-t'-\frac{2z_j}{c}  \right) -\sqrt{\gamma_{jR}\gamma_{jL}}\delta \left(t- t' + \frac{2z_j}{c}  \right), 
\end{aligned} 
\end{equation}
which is exactly Eq.~(\ref{con:Eqlemma1}).\qed
\end{proof}

\begin{lemma} \label{lemmabinjtTwoatom}
For the non-Markovian Hamiltonian in Eq.~(\ref{con:HeffSemiInputFieldKappa}), when $j\neq l$, the commutator between the $j$th and $l$th input quantum noise channels satisfies that
\begin{equation} \label{con:ttpexample}
\begin{aligned}
\left[ b_{\rm in}^{(j)} (t) ,{b_{\rm in}^{(l)}}^{\dag} (t')\right] 
=&\sqrt{\gamma_{jL}\gamma_{lL}}\delta\left( t - t' + \frac{z_j-z_l}{c}\right) +\sqrt{\gamma_{jR}\gamma_{lR}}\delta\left( t - t' + \frac{z_l-z_j}{c}\right) \\
&-\sqrt{\gamma_{jL}\gamma_{lR}} \delta\left( t - t' +\frac{z_j+z_l}{c}\right) -\sqrt{\gamma_{jR}\gamma_{lL}} \delta\left( t - t' -\frac{z_j+z_l}{c}\right).
\end{aligned}
\end{equation}
\end{lemma}

\begin{proof}
In Eq.~(\ref{con:bjrelation}) with different input noise channels,
\begin{equation} \label{con:bjltwoatom}
\begin{aligned}
& \left[ b_{\rm in}^{(j)} (t), {b_{\rm in}^{(l)}}^{\dag} (t') \right]\\
 =&\int_0^{t} \tilde{\kappa}_j^* (t-\nu)  b_{\rm in}\left (\nu\right)\mathrm{d}\nu \int_0^{t'} \tilde{\kappa}_l (t'-\nu')  b_{\rm in}^{\dag}\left (\nu'\right)\mathrm{d}\nu' -\int_0^{t'} \tilde{\kappa}_l (t'-\nu')  b_{\rm in}^{\dag}\left (\nu'\right)\mathrm{d}\nu' \int_0^{t} \tilde{\kappa}_j^* (t-\nu)  b_{\rm in}\left (\nu\right)\mathrm{d}\nu \\
=&\int_0^t \left[ \sqrt{\gamma_{jR}} \delta\left(t- \nu-\frac{z_j}{c}  \right) -\sqrt{\gamma_{jL}} \delta \left(t- \nu+\frac{z_j}{c}  \right) \right] \left[ \sqrt{\gamma_{lR}} \delta\left(t'- \nu-\frac{z_l}{c}  \right) -\sqrt{\gamma_{lL}} \delta \left(t'- \nu+\frac{z_l}{c}  \right) \right]\mathrm{d}\nu, 
\end{aligned}
\end{equation}
then Eq.~(\ref{con:ttpexample}) can be derived by similarly integrating Eq.~(\ref{con:bjltwoatom}) as in Eq.~(\ref{con:ttpexampleProof}) .\qed
\end{proof}

Based on Lemma~\ref{lemmabinjt} and  Lemma~\ref{lemmabinjtTwoatom}, we have the following proposition.
\begin{mypro} \label{SemiInfiniteJudgeNM}
In the quantum network in Fig.~\ref{fig:NatomWaveguide}(a), the quantum noise commutators  $\left[ b_{\rm in}^{(j)} (t) ,{b_{\rm in}^{(l)}}^{\dag} (t')\right]$ reduce to be Markovian  only when $N=1$ and $\gamma_{1L}\gamma_{1R} = 0$. 
\end{mypro}
\begin{proof}
When $N>1$, $\left[ b_{\rm in}^{(j)} (t) ,{b_{\rm in}^{(l)}}^{\dag} (t')\right] \neq \delta(t-t')$, and always contains at least one delayed component, according to Lemma~\ref{lemmabinjtTwoatom}. Thus in this case, the quantum noise commutators are non-Markovian. When $N=1$ and $\gamma_{1L}\gamma_{1R} = 0$, $\left[ b_{\rm in}^{(1)} (t) ,{b_{\rm in}^{(1)}}^{\dag} (t')\right] = \left(\gamma_{1R} + \gamma_{1L} \right) \delta\left(t-t' \right)$, according to Lemma~\ref{lemmabinjt}, the quantum noise commutator is Markovian. \qed
\end{proof}

\subsection{Non-Markovianity of quantum noises in an infinite waveguide}
When atoms are coupled to an infinite waveguide, akin to the approach in Appendix~\ref{Sec:AppdixTraceWaveguide}, the Hamiltonian can be reformulated by tracing over the waveguide Hamiltonian as
\begin{equation} \label{con:HeffInfiniteV2}
\begin{aligned}
\tilde{\mathbf{H}}(t)
 =&\sum_{j=1}^N \left(\omega_j^a - i\frac{\eta_j}{2}\right) \sigma_j^+\sigma_j^- \\
  &+ \sum_{j=1}^N \left[ i\sqrt{\gamma_{jR}} c_{\rm in}(t-z_j/c)\sigma_j^+  - i\sqrt{\gamma_{jR}} c_{\rm in}^{\dag}(t-z_j/c)\sigma_j^- \right]\\
 &+\sum_{j=1}^N \left[ i \sqrt{\gamma_{jL}}  d_{\rm in}(t+z_j/c)\sigma_j^+  -i \sqrt{\gamma_{jL}}  d_{\rm in}^{\dag}(t+z_j/c) \sigma_j^- \right],
\end{aligned}
\end{equation}
which is affected by the delayed right-propagating and left-propagating noise operators. Similar to the definition for semi-infinite waveguide case in Eq.~(\ref{con:Multichannle}), we denote
\begin{subequations} \label{con:InfiniteWaveJinput}
\begin{numcases}{}
{c_{\rm in}^{(j)}}^{\dag}(t) = \sqrt{\gamma_{jR}} c_{\rm in}^{\dag}(t-z_j/c),\label{cinjinputInf}\\
 {d_{\rm in}^{(j)}}^{\dag}(t)= \sqrt{\gamma_{jL}}  d_{\rm in}^{\dag}(t+z_j/c). \label{dinjinputInf}
\end{numcases}
\end{subequations}
Then Eq.~(\ref{con:HeffInfiniteV2}) can be rewritten as
\begin{equation} \label{con:HeffInfiniteV3}
\begin{aligned}
\tilde{\mathbf{H}}(t)
 =&\sum_{j=1}^N \left(\omega_j^a - i\frac{\eta_j}{2}\right)\sigma_j^+\sigma_j^- + i \sum_{j=1}^N  \left[   c_{\rm in}^{(j)}(t) \sigma_j^+ + d_{\rm in}^{(j)}(t) \sigma_j^+- \rm H.c. \right],
\end{aligned}
\end{equation}
where the quantum noise operators satisfy the following two lemmas.
\begin{lemma} \label{lemmadintInfinite}
For the non-Markovian Hamiltonian in Eq.~(\ref{con:HeffInfiniteV2}), the commutator for the $j$th input quantum noise channel satisfies
\begin{subequations} \label{con:InfiniteWaveCommutJ}
\begin{numcases}{}
 \left[ c_{\rm in}^{(j)} (t) ,{c_{\rm in}^{(j)}}^{\dag} (t')\right] = \gamma_{jR} \delta\left(t-t' \right),\label{cinjinputInfCommu}\\
 \left[ d_{\rm in}^{(j)} (t) ,{d_{\rm in}^{(j)}}^{\dag} (t')\right] = \gamma_{jL} \delta\left(t-t' \right). \label{dinjinputInfCommu}
\end{numcases}
\end{subequations}
\end{lemma}
\begin{proof}
The proof is similar to that of Lemma~\ref{lemmabinjt} and is therefore omitted. \qed
\end{proof}

\begin{lemma} \label{lemmabinjtTwoatomInfinte}
For the two input quantum noise channels with $j\neq l$ in the non-Markovian Hamiltonian in Eq.~(\ref{con:HeffInfiniteV2}), the commutators satisfy 
\begin{subequations} \label{con:InfiniteWaveCommutJTwoChannel}
\begin{numcases}{}
 \left[ c_{\rm in}^{(j)} (t) ,{c_{\rm in}^{(l)}}^{\dag} (t')\right] =  \sqrt{\gamma_{jR}\gamma_{lR}}\delta\left( t - t' + \frac{z_j-z_l}{c}\right),\label{Tworightchannel}\\
\left[ d_{\rm in}^{(j)} (t) ,{d_{\rm in}^{(l)}}^{\dag} (t')\right] =  \sqrt{\gamma_{jL}\gamma_{lL}}\delta\left( t - t' + \frac{z_j-z_l}{c}\right), \label{Twoleftchannel}
\end{numcases}
\end{subequations}
and for arbitrary $j$ and $l$,
\begin{equation} \label{con:LeftRightNoncommu}
\begin{aligned}
\left[d_{\rm in}^{(j)} (t),{c_{\rm in}^{(l)}}^{\dag} (t')\right] =0.
\end{aligned}
\end{equation}
\end{lemma}
\begin{proof}
The proof of Eq.~(\ref{con:InfiniteWaveCommutJTwoChannel}) is similar to that in Lemma~\ref{lemmabinjtTwoatom}, thus is omitted. Eq.~(\ref{con:LeftRightNoncommu}) holds because the left-propagation modes and the right-propagating modes are independent from each other in an infinite waveguide as $\left[l_{\omega} (\omega),r_{\omega}^{\dag} (\omega') \right] = 0$ in Eq.~(\ref{con:Hintj2nonchiral}). \qed
\end{proof}

\begin{mypro}  \label{InfiniteJudgeNM}
When $z_j \neq z_l$ for arbitrary two atoms, the quantum noise commutators in Eq.~(\ref{con:InfiniteWaveJinput}) reduce to be Markovian when $N=1$, or $N=2$ and $\gamma_{1R}\gamma_{2R} = \gamma_{1L}\gamma_{2L} = 0$. 
\end{mypro}
\begin{proof}
When $N=1$, the commutators of quantum noises are proportional to $\delta(t-t')$ according to Lemma~\ref{lemmadintInfinite}. When $N=2$, the quantum noise commutators are proportional to $\delta(t-t')$  only when $\gamma_{1R}\gamma_{2R} = \gamma_{1L}\gamma_{2L} = 0$ according to the right-hand side (RHS) of Eq.~(\ref{con:InfiniteWaveCommutJTwoChannel}) in Lemma~\ref{lemmabinjtTwoatomInfinte}. When $N>2$ atoms are coupled to the waveguide with $z_j \neq z_l$, there is always at least one delayed component in Eq.~(\ref{con:InfiniteWaveCommutJTwoChannel}), thus the dynamics is non-Markovian.
\qed
\end{proof}

The non-Markovian property of quantum noises in an infinite waveguide differs from that in a semi-infinite waveguide due to reflections by the terminal mirror of the semi-infinite waveguide, as noted in the following remark.
\begin{remark}
When one atom is simultaneously coupled to the left- and right-propagating modes of a semi-infinite waveguide, the dynamics must be non-Markovian. Conversely, when one atom is coupled to an infinite waveguide, the dynamics must be Markovian. When there are two atoms coupled to a semi-infinite waveguide, the dynamics must be non-Markovian. However, when two atoms are coupled to an infinite waveguide, the dynamics can be Markovian as in Proposition~\ref{InfiniteJudgeNM}. Furthermore, when more than two atoms are coupled to either an infinite or semi-infinite waveguide, the dynamics must be non-Markovian. Above all, the colored noise commutators with non-Markovian properties in Lemmas~\ref{lemmabinjt}-\ref{lemmabinjtTwoatomInfinte} are rooted in the non-Markovian interactions with different phases in Eqs.~(\ref{con:Hintjchiral},\ref{con:Hintj2nonchiral}), which can result in equivalent input quantum noises represented in the time domain with different time delays.
\end{remark}

\section{QSDE for atom-waveguide interactions} \label{Sec:QSDE}
Based on the Hamiltonian of the above quantum networks with stochastic input quantum noises, this section delves into the examination of non-Markovian coherent feedback dynamics based on quantum stochastic differential equations.

\subsection{QSDE for semi-infinite waveguide circumstance} \label{Sec:QSDEsemi}
The evolution of quantum states in  Fig.~\ref{fig:NatomWaveguide}(a) can be represented with the propagator $U(t)$ applied upon the initial state, such that $|\Psi(t)\rangle  = U(t)|\Psi(0)\rangle $, and the initial state $|\Psi(0)\rangle$ is assumed to be as follows.
\begin{assumption} \label{InitialState}
~\citep{fischer2018particle} Initially the waveguide is in its vacuum state as $|\mathbf{0}\rangle$, the initial atomic state is $|\psi(0)\rangle$, and $|\Psi(0)\rangle = |\psi(0)\rangle \otimes  |\mathbf{0}\rangle $.
\end{assumption}

The dynamics of $U(t)$ is governed by the Schr\"{o}dinger equation with the Hamiltonian in Eq.~(\ref{con:HeffSemiInputFieldKappa}) as~\citep{baragiola2012n}
\begin{equation} \label{con:Uequation}
\begin{aligned}
\mathrm{d}U(t) &= \left\{ -iH_s +\sum_{j=1}^N \left[  {b_{\rm in}^{(j)}}^{\dag} (t) L_j - b_{\rm in}^{(j)} (t) L_j^{\dag}\right]\right\}U(t)\mathrm{d} t\\
&\triangleq -i H_{\rm bs}U(t)\mathrm{d} t.
\end{aligned}
\end{equation}

Generalized from the Markovian circumstance~\citep{zhang2017quantum}, we define the following quantum stochastic processes, 
\begin{equation} \label{con:dBdef}
\begin{aligned}
B_{\rm in}^{(j)}(t) = \int_0^t b_{\rm in}^{(j)} (\tau)\mathrm{d} \tau,~{B_{\rm in}^{(j)}}^{\dag}(t) =\int_0^t {b_{\rm in}^{(j)}}^{\dag} (\tau)\mathrm{d} \tau,
\end{aligned}
\end{equation}
and we denote $\mathrm{d}B_{\rm in}^{(j)}$ and $\mathrm{d}{B_{\rm in}^{(j)}}^{\dag}$ as quantum noise differential elements.

\begin{remark} \label{ItoMarkovianRemark}
In the Markovian quantum stochastic dynamics, the quantum noise differential elements for Eq.~(\ref{con:dBdef}) reduce to $\mathrm{d} B$ and $\mathrm{d} B^{\dag}$, then the \rm{It\={o}} rule reads $\mathrm{d} B \mathrm{d} B^{\dag} =\mathrm{d}t $, and $ \mathrm{d} B^{\dag} \mathrm{d} B = \mathrm{d} B^{\dag} \mathrm{d} B^{\dag} =  \mathrm{d} B \mathrm{d} B= 0$~\citep{gardiner2004quantum,zhang2013non,li2022control,baragiola2012n}. 
\end{remark}

An extension of the above remark regarding the quantum noises in a non-Markovian network will be given below. Prior to that, it is necessary to clarify the dynamics of $U(t)$ in Eq.~(\ref{con:Uequation}).

Based on the quantum noise properties in Lemmas~\ref{lemmabinjt},~\ref{lemmabinjtTwoatom},~\ref{lemmadintInfinite} and~\ref{lemmabinjtTwoatomInfinte}, the non-Markovian equation for the propagator $U(t)$ reads 
\begin{equation} \label{con:dUtSDE2}
\begin{aligned}
\dot{U}(t) =& -iH_s  U(t) +\sum_{j=1}^N  {b_{\rm in}^{(j)}}^{\dag} (t) L_j U(t) - \sum_{j=1}^N  L_j^{\dag}U(t)b_{\rm in}^{(j)}(t)-  \sum_{j=1}^N \sum_{l=1}^N  L_j^{\dag}   \Gamma_{jl} (L,U,t),
\end{aligned}
\end{equation}
where we denote the gauge process of two noise channels as
\begin{equation} \label{con:GaugeSimple}
\begin{aligned}
\Gamma_{jl} (L,U,t)&= \int_0^t  \left [ b_{\rm in}^{(j)} (t),   {b_{\rm in}^{(l)}}^{\dag} (\tau) \right]L_lU(\tau)\mathrm{d} \tau\\
&= \int_0^t \int_0^t \tilde{\kappa}_j^* (t-\nu) \tilde{\kappa}_j (\tau-\nu) \mathrm{d}\nu L_lU(\tau)\mathrm{d} \tau,
\end{aligned}
\end{equation}
and
\begin{equation} \label{con:GammajlGauge}
\begin{aligned}
    &\Gamma_{jl} (L,U,t) =
   \begin{cases}
   \frac{\gamma_{jR} + \gamma_{jL}}{2}L_j U(t)-\sqrt{\gamma_{jR}\gamma_{jL}}L_j U\left(t-\frac{2z_j}{c}\right), &j = l,\\
    \sqrt{\gamma_{jL}\gamma_{lL}}L_l U \left(t- \frac{z_l-z_j}{c} \right) -\sqrt{ \gamma_{jR}\gamma_{lL}}L_lU \left(t- \frac{z_l+z_j}{c} \right), &j < l,\\
   \sqrt{\gamma_{jR}\gamma_{lR}} L_l U \left(t- \frac{z_j-z_l}{c} \right) - \sqrt{\gamma_{jR}\gamma_{lL}} L_l U \left(t- \frac{z_l+z_j}{c} \right), &j >l.
   \end{cases}
   \end{aligned}
  \end{equation}
See Appendix \ref{Sec:AppendixPropogator} for more details.

Based on Eq.~(\ref{con:GammajlGauge}), we can distinguish Markovian and non-Markovian stochastic quantum dynamics with arbitrary $z_j$ and $z_l$ according to the following proposition.

\begin{mypro}  \label{MarkovianCondition}
The quantum stochastic dynamics for semi-infinite waveguide-QED reduces to be Markovian when $\gamma_{jL}\gamma_{lR} = 0$ for arbitrary $j,l$, and $z_j = z_l$ for arbitrary  $j\neq l$ satisfying that $\gamma_{jR}\gamma_{lR} \neq 0$ or $\gamma_{jL}\gamma_{lL} \neq 0$.
\end{mypro}

\begin{proof}
When the conditions are satisfied, $\Gamma_{jl}$ in Eq.~(\ref{con:GammajlGauge}) is independent from the delayed components, thus the quantum stochastic dynamics reduces to be Markovian. \qed
\end{proof}

For an arbitrary operator $X$, we consider its dynamics in a tensor product format with the noise space $\mathbb{I}$ in the Heisenberg picture, and denote
\begin{equation} \label{con:blackXt}
\begin{aligned}
\mathbf{X}(t) = U(t)^{\dag} \left( X \otimes \mathbb{I}  \right)U(t),
\end{aligned}
\end{equation}
then~\citep{gough2008linear}
\begin{equation} \label{con:QSDEoperatorX}
\begin{aligned}
&\mathrm{d} \mathbf{X}  =\mathrm{d}U^{\dag} X U + U^{\dag} X \mathrm{d}U + \mathrm{d}U^{\dag} X \mathrm{d}U.
\end{aligned}
\end{equation}
According to Eq.~(\ref{con:dUtSDE2}), the first component on the RHS of Eq.~(\ref{con:QSDEoperatorX}) reads
\begin{equation} \label{con:QSDEoperatorXCal1}
\begin{aligned}
\mathrm{d}U^{\dag} X U = &i U^{\dag} H_s  X U\mathrm{d} t +\sum_{j=1}^N  U^{\dag} L_j ^{\dag} b_{\rm in}^{(j)}X U\mathrm{d} t  - \sum_{j=1}^N {b_{\rm in}^{(j)}}^{\dag}U^{\dag} L_jX U\mathrm{d} t  -\sum_{j=1}^N \sum_{l=1}^N   \Gamma_{jl}^{\dag} (L,U,t) L_j X U\mathrm{d} t\\
=&i H_{\rm bs} U^{\dag}   X U\mathrm{d} t - \sum_{j=1}^N \sum_{l=1}^N   \Gamma_{jl}^{\dag} (L,U,t) L_j X U\mathrm{d} t,
\end{aligned}
\end{equation}
where $\left[ b_{\rm in}^{(j)},U\right] = \left[H_{\rm bs},U\right] =0$ according to Eq.~(\ref{con:Uequation}).

Similarly, the second component of Eq.~(\ref{con:QSDEoperatorX}) reads
\begin{equation} \label{con:QSDEoperatorXCal2}
\begin{aligned}
&U^{\dag} X\mathrm{d} U 
=-i U^{\dag} XU H_{\rm bs}\mathrm{d} t- \sum_{j=1}^N \sum_{l=1}^N U^{\dag} X  L_j^{\dag}   \Gamma_{jl} (L,U,t)\mathrm{d} t.
\end{aligned}
\end{equation}
According to Remark~\ref{ItoMarkovianRemark} and Eq.~(\ref{con:dUtSDE2}), the third component on the RHS of Eq.~(\ref{con:QSDEoperatorX}) reads
\begin{equation} \label{con:QSDEoperatorXCal3}
\begin{aligned}
\mathrm{d}U^{\dag} X \mathrm{d}U  
= \sum_{j=1}^N  \sum_{l=1}^N \mathrm{d} B_{\rm in}^{(j)}{\mathrm{d} B_{\rm in}^{(l)}}^{\dag} U^{\dag} L_j ^{\dag}  X   L_lU,
\end{aligned}
\end{equation}
where $\mathrm{d} B_{\rm in}^{(j)} {\mathrm{d} B_{\rm in}^{(l)}}^{\dag}$ can be clarified by the following theorem on the {\rm{It\={o}}} rule for quantum noises in a non-Markovian waveguide-QED network.

\begin{Theorem} \label{ItoSemi}
Based on Assumption~\ref{InitialState} and $\mathrm{d}t$ is small such that $\mathrm{d}t< z_1/c$, the non-Markovian interaction between a semi-infinite waveguide and multi-atom network in Fig.~\ref{fig:NatomWaveguide}(a) is equivalent to the quantum stochastic network with multiple input noisy channels in Eq.~(\ref{con:dBdef}) satisfying the {\rm{It\={o}}} rule 
\begin{equation} \label{con:NonMarkovianIto2}
\begin{aligned} 
\mathrm{d} B_{\rm in}^{(j)}(t) {\mathrm{d} B_{\rm in}^{(l)}}^{\dag}(t) =& \int_{t}^{t+\mathrm{d}t}\int_{t}^{t+\mathrm{d}t}  \left[b_{\rm in}^{(j)}(s), {b_{\rm in}^{(l)}}^{\dag}(\mu)\right]\mathrm{d}s\mathrm{d}\mu\\
=&  \begin{cases}
\left(\gamma_{jL}+\gamma_{jR}\right) \mathrm{d}t, & j = l,\\
 \left( \sqrt{\gamma_{jL}\gamma_{lL}} + \sqrt{\gamma_{jR}\gamma_{lR}} \right)\mathrm{d}t  , &j\neq l,\frac{\left|z_j-z_l\right|}{c} \leq \mathrm{d}t,\\
 0 , &j\neq l, \frac{\left|z_j-z_l\right|}{c}> \mathrm{d}t,
   \end{cases}
\end{aligned}
\end{equation}
with the commutator given by Eq.~(\ref{con:bjrelation}) or Lemmas~\ref{lemmabinjt},\ref{lemmabinjtTwoatom}, and $\mathrm{d} B_{\rm in}^{(j)} \mathrm{d} B_{\rm in}^{(l)} = {\mathrm{d} B_{\rm in}^{(j)}}^{\dag} {\mathrm{d} B_{\rm in}^{(l)}}^{\dag} ={\mathrm{d} B_{\rm in}^{(l)}}^{\dag} \mathrm{d} B_{\rm in}^{(j)} = 0$.
\end{Theorem}
\begin{proof}
Based on the definition of $B_{\rm in}^{(j)}(t)$ in Eq.~(\ref{con:dBdef}), $\mathrm{d} B_{\rm in}^{(j)} \mathrm{d} B_{\rm in}^{(l)} = {\mathrm{d} B_{\rm in}^{(j)}}^{\dag} {\mathrm{d} B_{\rm in}^{(l)}}^{\dag} ={\mathrm{d} B_{\rm in}^{(l)}}^{\dag} \mathrm{d} B_{\rm in}^{(j)} = 0$ because $\mathrm{d} B_{\rm in}^{(j)}$ annihilates the quantum noise in the waveguide, and this is the same as the Markovian case in Remark~\ref{ItoMarkovianRemark}. Then according to Lemmas~\ref{lemmabinjt},\ref{lemmabinjtTwoatom},
\begin{equation} 
\begin{aligned} \label{con:ItoNewcal}
&\left[ \mathrm{d} B_{\rm in}^{(j)}(t), {\mathrm{d} B_{\rm in}^{(l)}}^{\dag}(t)\right] \\
=& \int_{t}^{t+\mathrm{d}t}\int_{t}^{t+\mathrm{d}t} \left[b_{\rm in}^{(j)}(s), {b_{\rm in}^{(l)}}^{\dag}(\mu)\right]\mathrm{d}s\mathrm{d}\mu\\
=&\int_{t}^{t+\mathrm{d}t}\int_{t}^{t+\mathrm{d}t} \left[\sqrt{\gamma_{jL}\gamma_{lL}}\delta\left( s - \mu + \frac{z_j-z_l}{c}\right)+\sqrt{\gamma_{jR}\gamma_{lR}} \delta\left( s - \mu + \frac{z_l-z_j}{c}\right) \right.\\
&\left.-\sqrt{\gamma_{jL}\gamma_{lR}} \delta\left( s - \mu +\frac{z_j+z_l}{c}\right)  -\sqrt{\gamma_{jR}\gamma_{lL}} \delta\left(s - \mu -\frac{z_j+z_l}{c}\right)\right]\mathrm{d}s\mathrm{d}\mu,
\end{aligned}
\end{equation}
where $s,\mu \in \left[t,t+\mathrm{d}t \right]$. For the parameter setting that $\left|z_j-z_l\right|/c > \mathrm{d}t$ for arbitrary $j \neq l$, Eq.~(\ref{con:ItoNewcal}) has a non-zero value only when $j =l$, and 
\begin{equation} 
\begin{aligned} \label{con:ItoNewcal2}
\left[ \mathrm{d} B_{\rm in}^{(j)}(t), {\mathrm{d} B_{\rm in}^{(j)}}^{\dag}(t)\right] =&\int_{t}^{t+\mathrm{d}t}\int_{t}^{t+\mathrm{d}t} \left(\gamma_{jL}+\gamma_{jR}\right)\delta\left( s - \mu\right)\mathrm{d}s\mathrm{d}\mu\\
=&\left(\gamma_{jL}+\gamma_{jR}\right) \int_{t}^{t+\mathrm{d}t}\mathrm{d}s\\
=&\left(\gamma_{jL}+\gamma_{jR}\right) \mathrm{d}t.
\end{aligned}
\end{equation}
Besides, when there exist $j\neq l$ satisfying that $\left|z_j-z_l \right|/c \leq \mathrm{d}t$, then
\begin{equation} 
\begin{aligned} \label{con:ItoNewcal3}
\left[ \mathrm{d} B_{\rm in}^{(j)}(t), {\mathrm{d} B_{\rm in}^{(l)}}^{\dag}(t)\right] =&\int_{t}^{t+\mathrm{d}t}\int_{t}^{t+\mathrm{d}t} \left[\sqrt{\gamma_{jL}\gamma_{lL}}\delta\left( s - \mu + \frac{z_j-z_l}{c}\right) +\sqrt{\gamma_{jR}\gamma_{lR}} \delta\left( s - \mu + \frac{z_l-z_j}{c}\right) \right]\mathrm{d}s\mathrm{d}\mu\\
=& \left( \sqrt{\gamma_{jL}\gamma_{lL}} + \sqrt{\gamma_{jR}\gamma_{lR}} \right)\mathrm{d}t.
\end{aligned}
\end{equation}
Combined with ${\mathrm{d} B_{\rm in}^{(l)}}^{\dag} \mathrm{d} B_{\rm in}^{(j)} = 0$, Eq.~(\ref{con:NonMarkovianIto2}) can be derived. \qed
\end{proof}

Using the quantum noise properties above, Eq.~(\ref{con:QSDEoperatorX}) can be equivalently written as a non-Markovian QSDE by summarizing Eqs.~(\ref{con:QSDEoperatorXCal1},\ref{con:QSDEoperatorXCal2},\ref{con:QSDEoperatorXCal3}) as follows.

According to Eq.~(\ref{con:bjrelation}), Lemma~\ref{lemmabinjt}, Lemma~\ref{lemmabinjtTwoatom} and Theorem~\ref{ItoSemi}, $\left[ b_{\rm in}^{(j)} (t) ,{b_{\rm in}^{(l)}}^{\dag} (\tau)\right]$ is a function of $(t-\tau)$, by defining 
\begin{equation} \label{con:QSDEKernel0}
\begin{aligned}
\bar{\kappa}_{jl}^{\rm b}(t-\tau) &=\left [ b_{\rm in}^{(j)} (t), {b_{\rm in}^{(l)}}^{\dag} (\tau) \right],
\end{aligned}
\end{equation}
then in Eq.~(\ref{con:GaugeSimple}),
\begin{equation} \label{con:Gaugekernel}
\begin{aligned}
\Gamma_{jl} (L,U,t)&= \int_0^t  \bar{\kappa}_{jl}^{\rm b}(t-\tau)L_lU(\tau)\mathrm{d} \tau.
\end{aligned}
\end{equation}
Take the last component in Eq.~(\ref{con:QSDEoperatorXCal2})  as an example, 
\begin{equation} \label{con:chijlLXLcal}
\begin{aligned}
 U^{\dag} X  L_j^{\dag}   \Gamma_{jl} (L,U,t) =&U^{\dag}(t)XU(t) U^{\dag}(t) L_j^{\dag}  U(t)U^{\dag}(t)   \int_0^t  \bar{\kappa}_{jl}^{\rm b}(t-\tau)L_lU(\tau)\mathrm{d} \tau\\
=&\mathbf{X}(t)\mathbf{L}_j^{\dag}(t)   \int_0^t  \bar{\kappa}_{jl}^{\rm b}(t-\tau)U^{\dag}(t-\tau) \mathbf{L}_l(\tau)\mathrm{d} \tau,
\end{aligned}
\end{equation}
where $\mathbf{X}(t)$ is defined in Eq.~(\ref{con:blackXt}),  $\mathbf{L}_j(t) = U^{\dag}(t) L_j U(t)$, $U^{\dag}(t-\tau) = e^{iH_{\rm bs}(t-\tau)}$, and we further denote
\begin{equation} \label{con:QSDEKernel}
\begin{aligned}
\kappa_{jl}^{\rm b}(t-\tau) &= \left [ b_{\rm in}^{(j)} (t), {b_{\rm in}^{(l)}}^{\dag} (\tau) \right]e^{iH_{\rm bs}(t-\tau)},
\end{aligned}
\end{equation}
thus
\begin{equation} \label{con:chijlLXLcal2}
\begin{aligned}
 U^{\dag} X  L_j^{\dag}   \Gamma_{jl} (L,U,t) &=\mathbf{X}(t)\mathbf{L}_j^{\dag}(t)   \int_0^t \kappa_{jl}^{\rm b}(t-\tau) \mathbf{L}_l(\tau)\mathrm{d} \tau.
\end{aligned}
\end{equation}

Similarly, the last component in Eq.~(\ref{con:QSDEoperatorXCal1}) reads 
\begin{equation} 
\begin{aligned}
 \Gamma_{jl}^{\dag} (L,U,t) L_j X U = \int_0^t \mathbf{L}_l^{\dag}(\tau)   \kappa_{jl}^{\rm b\dag}(t-\tau)  \mathrm{d} \tau \mathbf{L}_j(t)  \mathbf{X}(t),
\end{aligned}
\end{equation}
and the component $\mathrm{d}U^{\dag} X \mathrm{d}U$ in Eq.~(\ref{con:QSDEoperatorX}) is given as the following lemma.
\begin{lemma}
For the third component on the RHS of Eq.~(\ref{con:QSDEoperatorX}),
\begin{equation} \label{con:ItoLindblad}
\begin{aligned}
\mathrm{d}U^{\dag} X \mathrm{d}U = &\left[\sum_{j,l=1}^N \mathbf{L}_j^{\dag}(t) \mathbf{X}(t) \int_0^t \kappa_{jl}^{\rm b}(t-\tau) \mathbf{L}_l(\tau)\mathrm{d} \tau +\sum_{j,l=1}^N \int_0^t  \mathbf{L}_l^{\dag}(\tau)\kappa_{jl}^{\rm b\dag}(t-\tau)\mathrm{d} \tau \mathbf{X}(t) \mathbf{L}_j(t)  \right] \mathrm{d}t.
\end{aligned}
\end{equation}
\end{lemma}

\begin{proof}
Considering that $b_{\rm in}^{(j)}(t)\mathrm{d}t =\mathrm{d}B_{\rm in}^{(j)}$, according to Eq.~(\ref{con:QSDEoperatorXCal3}), $\mathrm{d}U^{\dag} X \mathrm{d}U$ can be simplified after omitting the higher-order infinitesimal of $\mathrm{d}t$ as
\begin{equation} 
\begin{aligned}
 \mathrm{d}U^{\dag} X \mathrm{d}U
=& \sum_{j=1}^N  U^{\dag}(t) L_j ^{\dag}(t)  \mathrm{d} B_{\rm in}^{(j)}(t) X(t)  \sum_{l=1}^N   {\mathrm{d} B_{\rm in}^{(l)}}^{\dag}(\tau) L_l(\tau)U(\tau) + \sum_{j=1}^N  U^{\dag}(\tau) L_j ^{\dag}(\tau)  \mathrm{d} B_{\rm in}^{(j)}(\tau) X(t)  \sum_{l=1}^N   {\mathrm{d} B_{\rm in}^{(l)}}^{\dag}(t) L_l(t)U(t)\\
=& \sum_{j,l=1}^N \mathrm{d} B_{\rm in}^{(j)}(t){\mathrm{d} B_{\rm in}^{(l)}}^{\dag}(\tau) U^{\dag}(t) L_j ^{\dag} (t) X(t)   L_l(\tau)U(\tau) + \sum_{j,l=1}^N \mathrm{d} B_{\rm in}^{(j)}(\tau)  {\mathrm{d} B_{\rm in}^{(l)}}^{\dag}(t) U^{\dag}(\tau) L_j ^{\dag}(\tau)  X(t)   L_l(t)U(t)\\
=& \sum_{j,l=1}^N \int_0^t  \bar{\kappa}_{jl}^{\rm b}(t-\tau) U^{\dag}(t) L_j ^{\dag} (t) X(t)   L_l(\tau)U(\tau)\mathrm{d} \tau\mathrm{d}t + \sum_{j,l=1}^N \int_0^t  \bar{\kappa}_{jl}^{\rm b}(\tau-t)  U^{\dag}(\tau) L_j ^{\dag}(\tau)  X(t)   L_l(t)U(t)\mathrm{d} \tau \mathrm{d}t\\
=& \sum_{j,l=1}^N \mathbf{L}_j^{\dag}(t)  \mathbf{X}(t) \int_0^t  \kappa_{jl}^{\rm b}(t-\tau)  \mathbf{L}_l(\tau)\mathrm{d} \tau\mathrm{d}t + \sum_{j,l=1}^N \int_0^t    \mathbf{L}_j ^{\dag}(\tau) \kappa_{jl}^{\rm b\dag}(t-\tau)\mathrm{d} \tau  \mathbf{X}(t)   \mathbf{L}_l(t)\mathrm{d}t.
\end{aligned}
\end{equation}
The proof above is similar to the derivations in Eqs.~(\ref{con:chijlLXLcal},\ref{con:chijlLXLcal2}), or can be directly derived based on Eqs.~(\ref{con:chijlLXLcal},\ref{con:chijlLXLcal2}) according to the general mathematical format of the Lindblad component for a master equation~\citep{gorini1976completely,lindblad1976generators} as well as the non-Markovian Lindblad format in~\citep{zhang2013non}. 
\qed
\end{proof}

Then the non-Markovian QSDE for a multi-atom network coupled to a semi-infinite waveguide can be simplified as
\begin{small}
\begin{equation} \label{con:QSDESemiKernal}
\begin{aligned}
\mathrm{d} \mathbf{X}  = & -i \left [ \mathbf{X},  H_s \right]\mathrm{d}t +   \sum_{j=1}^N  \mathrm{d} {B_{\rm in}^{(j)}}^{\dag}  \left [ \mathbf{X}, L_j  \right] -  \sum_{j=1}^N \left [ \mathbf{X}, L_j^{\dag}  \right]\mathrm{d}B_{\rm in}^{(j)}\\
&+\left(\sum_{j,l=1}^N \left[\mathbf{L}_j^{\dag}(t) ,\mathbf{X}(t)\right]  \int_0^t \kappa_{jl}^{\rm b}(t-\tau) \mathbf{L}_l(\tau)\mathrm{d} \tau \right)\mathrm{d}t +\left(\sum_{j,l=1}^N  \int_0^t  \mathbf{L}_l^{\dag}(\tau) \kappa_{jl}^{\rm b\dag}(t-\tau)\mathrm{d} \tau  \left[\mathbf{X}(t),\mathbf{L}_j(t)\right] \right)\mathrm{d}t,
\end{aligned}
\end{equation}
\end{small}%
which is induced by the non-Markovian commutative relationships for quantum noises, and is different from the Markovian cases, such as in~\citep{yamamoto2012pure}.

\subsection{QSDE for infinite waveguide circumstance} 
The evolution of the quantum states in  Fig.~\ref{fig:NatomWaveguide}(b) can be described using the propagator $\tilde{U}(t)$, thus the quantum state evolves as $|\tilde{\Psi}(t)\rangle  = \tilde{U}(t)|\tilde{\Psi}(0)\rangle $. Referring to the state representation in Sec.~\ref{Sec:InputOutputInfiW}, we assume that the initial state of the system is as follows.
\begin{assumption} \label{InitialStateInfiniteWave}
~\citep{fischer2018particle} Initially the waveguide is in its vacuum state, namely $|\tilde{\Psi}(0)\rangle = |\tilde{\psi}(0)\rangle \otimes  |\mathbf{0}\rangle \otimes  |\mathbf{0}\rangle$.
\end{assumption}

The QSDE describing the interactions between the atom network and an infinite waveguide can be simplified from the semi-infinite case by segregating the waveguide modes into left-propagating and right-propagating modes. For example, setting the parameters as $\gamma_{jL} = 0$ and $\gamma_{jR} \neq 0$ for arbitrary $j$ in Fig.~\ref{fig:NatomWaveguide}(a), corresponds to a situation where atoms solely interact with the right-propagating modes in the waveguide, akin to the configuration depicted in  Fig.~\ref{fig:NatomWaveguide}(b) with the same parameter settings. Consequently, the general dynamics in Fig.~\ref{fig:NatomWaveguide}(b) can be inferred from this simplified setting, by further taking $\gamma_{jL} \neq 0$ in a similar approach as the right-propagating modes.  Hence, the detailed derivation on QSDE here can be omitted.

Thus generalized from Eq.~(\ref{con:Uequation}), the dynamics of $\tilde{U}(t)$ reads
\begin{small}
\begin{equation} \label{con:UequationInfinite}
\begin{aligned}
\mathrm{d}\tilde{U} &= \left\{ -iH_s +\sum_{j=1}^N \left[  \left({c_{\rm in}^{(j)}}^{\dag}  + {d_{\rm in}^{(j)}}^{\dag} \right) L_j-\rm H.c.\right]\right\}\tilde{U}\mathrm{d} t,
\end{aligned}
\end{equation}
\end{small}%
where generalized from Eq.~(\ref{con:dBdef}), we denote 
\begin{equation} \label{con:dCdefine}
\begin{aligned}
C_{\rm in}^{(j)}(t) = \int_0^t c_{\rm in}^{(j)} (\tau)\mathrm{d} \tau,~{C_{\rm in}^{(j)}}^{\dag}(t) = \int_0^t {c_{\rm in}^{(j)}}^{\dag} (\tau)\mathrm{d} \tau,
\end{aligned}
\end{equation}
and
\begin{equation} \label{con:dDdefine}
\begin{aligned}
D_{\rm in}^{(j)}(t) =\int_0^t d_{\rm in}^{(j)} (\tau)\mathrm{d} \tau,~{D_{\rm in}^{(j)}}^{\dag}(t) = \int_0^t {d_{\rm in}^{(j)}}^{\dag} (\tau)\mathrm{d} \tau.
\end{aligned}
\end{equation}
Then
\begin{equation} \label{con:dUtSDEInfinite}
\begin{aligned}
&\dot{\tilde{U}}(t) =  -iH_s \tilde{U}(t) + \sum_{j=1}^N  \left({c_{\rm in}^{(j)}}^{\dag}(t) + {d_{\rm in}^{(j)}}^{\dag}(t) \right)L_j \tilde{U}(t) -\sum_{j=1}^N  L_j^{\dag}\left(c_{\rm in}^{(j)}(t) + d_{\rm in}^{(j)}(t)\right)\tilde{U}(t)-\sum_{j=1}^N \sum_{l=1}^N  L_j^{\dag}   \tilde{\Gamma}_{jl} \left(L,\tilde{U},t\right),
\end{aligned}
\end{equation}
where
\begin{equation} \label{con:GammajlGauge2}
    \tilde{\Gamma}_{jl} \left(L,\tilde{U},t\right)=
   \begin{cases}
   \frac{\gamma_{jR} + \gamma_{jL}}{2}L_j \tilde{U}(t), &j = l,\\
   \sqrt{ \gamma_{jL}\gamma_{lL}}L_l\tilde{U}\left(t- \frac{z_l-z_j}{c} \right) , &j < l,\\
   \sqrt{\gamma_{jR}\gamma_{lR}} L_l \tilde{U} \left(t- \frac{z_j-z_l}{c} \right), &j >l.
   \end{cases}
\end{equation}
In analogy to Proposition~\ref{MarkovianCondition}, Eq.~(\ref{con:GammajlGauge2}) shows that the dynamics of the propagator $\tilde{U}(t)$ in Eq.~(\ref{con:dUtSDEInfinite}) reduces to be Markovian only if $\gamma_{jL}\gamma_{lL} = \gamma_{jR}\gamma_{lR} = 0$ whenever $z_j \neq z_l$. Similar to Eq.~(\ref{con:Gaugekernel}), we rewrite $\tilde{\Gamma}_{jl} (L,U,t)$ in the format of an integral with the non-Markovian kernel as
\begin{equation} \label{con:GammadeltaInfinite}
\begin{aligned}
  \tilde{\Gamma}_{jl} (L,U,t) = \int_0^t  \bar{\kappa}_{jl}^{\rm cd}(t-\tau)L_l\tilde{U}(\tau)  \mathrm{d} \tau,
\end{aligned}
\end{equation}
where 
\begin{equation}
\begin{aligned}
&\bar{\kappa}_{jl}^{\rm cd}(t-\tau) 
= \begin{cases}
     \left [ c_{\rm in}^{(j)} (t), {c_{\rm in}^{(l)}}^{\dag} (\tau) \right] + \left [ d_{\rm in}^{(j)} (t), {d_{\rm in}^{(l)}}^{\dag} (\tau) \right], &j = l,\\
    \left [ d_{\rm in}^{(j)} (t), {d_{\rm in}^{(l)}}^{\dag} (\tau) \right] , &j < l,\\
   \left [ c_{\rm in}^{(j)} (t), {c_{\rm in}^{(l)}}^{\dag} (\tau) \right] , &j >l,
   \end{cases}\notag
\end{aligned}
\end{equation}
and the detailed mathematical formats of $\left [ d_{\rm in}^{(j)} (t), {d_{\rm in}^{(l)}}^{\dag} (\tau) \right]$ and $ \left [ c_{\rm in}^{(j)} (t), {c_{\rm in}^{(l)}}^{\dag} (\tau) \right] $ are given by Lemma~\ref{lemmadintInfinite} and Lemma~\ref{lemmabinjtTwoatomInfinte}.

Then we can derive the non-Markovian \rm{It\={o}} rule for the quantum noises in the infinite waveguide as the following theorem.
\begin{Theorem} \label{ItoInfinite}
Based on Assumption~\ref{InitialStateInfiniteWave}, the non-Markovian interaction between an infinite waveguide and multi-atom network in Fig.~\ref{fig:NatomWaveguide}(b) is equivalent to the quantum stochastic network with multiple input noisy channels in Eqs.~(\ref{con:dCdefine},\ref{con:dDdefine}) satisfying the {\rm{It\={o}}} rule
\begin{equation} 
\begin{aligned} \label{con:NonMarkovianItoC}
\mathrm{d} C_{\rm in}^{(j)}(t) {\mathrm{d} C_{\rm in}^{(l)}}^{\dag}(t) =& \int_{t}^{t+\mathrm{d}t}\int_{t}^{t+\mathrm{d}t}  \left[c_{\rm in}^{(j)}(s), {c_{\rm in}^{(l)}}^{\dag}(\mu)\right]\mathrm{d}s\mathrm{d}\mu\\
=&  \begin{cases}
\gamma_{jR} \mathrm{d}t, & j = l,\\
\sqrt{\gamma_{jR}\gamma_{lR}}\mathrm{d}t  , &j\neq l,\frac{\left|z_j-z_l\right|}{c} \leq \mathrm{d}t,\\
 0 , &j\neq l, \frac{\left|z_j-z_l\right|}{c}> \mathrm{d}t,
   \end{cases}
\end{aligned}
\end{equation}
\begin{equation} 
\begin{aligned} \label{con:NonMarkovianItoD}
\mathrm{d} D_{\rm in}^{(j)}(t) {\mathrm{d} D_{\rm in}^{(l)}}^{\dag}(t) =& \int_{t}^{t+\mathrm{d}t}\int_{t}^{t+\mathrm{d}t}  \left[d_{\rm in}^{(j)}(s), {d_{\rm in}^{(l)}}^{\dag}(\mu)\right]\mathrm{d}s\mathrm{d}\mu\\
=&  \begin{cases}
\gamma_{jL} \mathrm{d}t, & j = l,\\
\sqrt{\gamma_{jL}\gamma_{lL}}\mathrm{d}t  , &j\neq l,\frac{\left|z_j-z_l\right|}{c} \leq \mathrm{d}t,\\
 0 , &j\neq l, \frac{\left|z_j-z_l\right|}{c} > \mathrm{d}t,
   \end{cases}
\end{aligned}
\end{equation}
$\mathrm{d} C_{\rm in}^{(j)} \mathrm{d} C_{\rm in}^{(l)} = {\mathrm{d} C_{\rm in}^{(j)}}^{\dag} {\mathrm{d} C_{\rm in}^{(l)}}^{\dag} ={\mathrm{d} C_{\rm in}^{(l)}}^{\dag} \mathrm{d} C_{\rm in}^{(j)} = 0$ and $\mathrm{d} D_{\rm in}^{(j)} \mathrm{d} D_{\rm in}^{(l)} = {\mathrm{d} D_{\rm in}^{(j)}}^{\dag} {\mathrm{d} D_{\rm in}^{(l)}}^{\dag} ={\mathrm{d} D_{\rm in}^{(l)}}^{\dag} \mathrm{d} D_{\rm in}^{(j)} = 0$.
\end{Theorem}
\begin{proof}
The proof is similar to that of Theorem~\ref{ItoSemi} based on the commutators in Lemmas~\ref{lemmadintInfinite},\ref{lemmabinjtTwoatomInfinte}, thus more details are omitted. \qed
\end{proof}

\subsection{Non-Markovian dynamics for single-point coupling}  \label{generalMaster}
To enhance the paper's conciseness and clarify the quantum measurement and filtering processes in non-Markovian quantum networks, we summarize the dynamics for both the semi-infinite waveguide case depicted in Fig.~\ref{fig:NatomWaveguide}(a) and the infinite waveguide case shown in Fig.~\ref{fig:NatomWaveguide}(b) using a general stochastic model in this section. Subsequently, the master equations can be derived using the aforementioned kernels.

For the most common scenario, the Hamiltonian describing an atom network coupled through a waveguide can be expressed as
\begin{equation} \label{con:generalHam}
\begin{aligned}
H = H_s + H_d,
\end{aligned}
\end{equation}
where $H_s$ is the free atomic Hamiltonian in Eq.~(\ref{con:Hsdef}) both for the semi-infinite and infinite waveguide cases, and $H_d$ is the controlled Hamiltonian by applying arbitrary control fields upon the atoms. Generalized from the simplified Markovian circumstance in~\citep{fischer2018particle}, the effective Hamiltonian can be defined as the \rm{It\={o}} format according to the first component of Eq.~(\ref{con:GammajlGauge}) and Eq.~(\ref{con:GammajlGauge2}) with $j=l$ as
\begin{equation} \label{con:HeffDefine}
\begin{aligned}
H_{\rm eff} = H - i\sum_j \frac{\gamma_{jR} + \gamma_{jL}}{2} L_j^{\dag} L_j,
\end{aligned}
\end{equation}
where the second component arises from the \rm{It\={o}} rule for quantum noises in the waveguide.

The QSDE governing the propagator $\mathbb{U}$ can be concluded based on Eqs.~(\ref{con:dUtSDE2},\ref{con:dUtSDEInfinite}) as 
\begin{equation} \label{con:QSDEPropagatorGeneralFormat}
\begin{aligned}
\mathrm{d} \mathbb{U}  
=&-iH_{\rm eff} \mathbb{U} \mathrm{d}t +  \sum_{j=1}^N  \mathrm{d}{\mathfrak{B}_{\rm in}^{(j)}}^{\dag} L_j \mathbb{U}  -  \sum_{j=1}^N L_j^{\dag} \mathbb{U}   \mathrm{d}\mathfrak{B}_{\rm in}^{(j)} +\sum_{j,l=1}^N  \int_0^t  L_j^{\dag}  \mathbf{f}_{jl}^{\delta} \left(t-\tau\right) L_l \mathbb{U}  (\tau)  \mathrm{d} \tau \mathrm{d}t,  
\end{aligned}
\end{equation}
where $\mathrm{d} \mathfrak{B}_{\rm in}^{(j)} = \mathrm{d} B_{\rm in}^{(j)}$ for the semi-infinite waveguide in Fig.~\ref{fig:NatomWaveguide}(a), $\mathrm{d} \mathfrak{B}_{\rm in}^{(j)} = \mathrm{d} C_{\rm in}^{(j)} + \mathrm{d} D_{\rm in}^{(j)}$ for the infinite waveguide in Fig.~\ref{fig:NatomWaveguide}(b), and we denote  $\mathfrak{b}_{\rm in}^{(j)} = \mathrm{d} \mathfrak{B}_{\rm in}^{(j)}/\mathrm{d}t$, 
\begin{subequations} \label{con:fjldeltaSemi}
\begin{numcases}{}
\mathbf{f}_{jl}\left(t-\tau\right) = \left[\mathfrak{b}_{\rm in}^{(j)} (t),{\mathfrak{b}_{\rm in}^{(l)}}^{\dag} (\tau)\right],\\
\mathbf{f}_{jl}^{\delta} \left(t-\tau\right) = \mathbf{f}_{jl}\left(t-\tau\right) -  \frac{\gamma_{jR} + \gamma_{jL}}{2}\delta_{jl} \delta(0),
\end{numcases}
\end{subequations}
for both the semi-infinite and infinite waveguide cases with $\delta_{jl}=1$ when $j=l$ and $\delta_{jl}=0$ when $j\neq l$.

However, according to Theorems~\ref{ItoSemi},\ref{ItoInfinite} on {\rm{It\={o}}} rules, the QSDE~(\ref{con:QSDEPropagatorGeneralFormat}) with the discrete time format can be simplified with the effective Hamiltonian
\begin{equation} \label{con:HeffDefineSimple}
\begin{aligned}
H_{\rm eff} = &H- i \sum_{\left|z_j-z_l \right|/c < \mathrm{d}t} \frac{\sqrt{\gamma_{jR}\gamma_{lR}}+\sqrt{\gamma_{jL}\gamma_{lL}}}{2}L_j^{\dag} L_l,
\end{aligned}
\end{equation}
thus when $0\leq \left|z_j-z_l \right|/c < \mathrm{d}t$,
\begin{equation}
\begin{aligned}
\mathbf{f}_{jl}^{\delta} \left(t-\tau\right) =& \mathbf{f}_{jl}\left(t-\tau\right) - \frac{\sqrt{\gamma_{jR}\gamma_{lR}}+\sqrt{\gamma_{jL}\gamma_{lL}}}{2}\delta\left(t-\frac{z_j-z_l}{c}\right).\notag
\end{aligned}
\end{equation}
This can also be interpreted from the perspective of Markovian approximations with $z_j-z_l \approx 0$ in Sec.~\ref{MarkovianSubsection}.

Above all, the QSDE for an arbitrary operator $\mathbb{X}$ reads
\begin{small}
\begin{equation} \label{con:QSDEGeneralFormat}
\begin{aligned}
\mathrm{d} \mathbb{X}  = & -i \left [  \mathbb{X},  H  \right]\mathrm{d}t +  \sum_{j=1}^N  \mathrm{d} {\mathfrak{B}_{\rm in}^{(j)}}^{\dag} \left [  \mathbb{X}, L_j  \right]-  \sum_{j=1}^N \left [  \mathbb{X}, L_j^{\dag}  \right]\mathrm{d}\mathfrak{B}_{\rm in}^{(j)}\\
&+\sum_{j,l=1}^N  \left( \left[\mathbf{L}_j^{\dag}(t),\mathbb{X} \right]\int_0^t \mathbf{\mathfrak{L}}_{jl}(t-\tau) \mathbf{L}_l(\tau) \mathrm{d} \tau \right) \mathrm{d} t +\sum_{j,l=1}^N  \left( \int_0^t  \mathbf{L}_l^{\dag}(\tau)\mathbf{\mathfrak{L}}_{jl}^{\dag}(t-\tau) \mathrm{d} \tau  \left[\mathbb{X},\mathbf{L}_j(t) \right] \right) \mathrm{d} t,
\end{aligned}
\end{equation}
\end{small}%
where we denote
\begin{subequations} \label{con:QSDEKernelSemiInfinite}
\begin{numcases}{}
\mathbf{\mathfrak{L}}_{jl}(t-\tau) = \left [{\mathfrak{b}}_{\rm in}^{(j)} (t), {{\mathfrak{b}_{\rm in}^{(l)}}}^{\dag} (\tau) \right]e^{i\tilde{H}_s(t-\tau)},\\
\tilde{H}_s(t) = H_s + i\sum_{j=1}^N \left[  \mathrm{d} {{\mathfrak{B}}_{\rm in}^{(j)}}^{\dag}(t) L_j - L_j^{\dag} \mathrm{d} {\mathfrak{B}}_{\rm in}^{(j)}(t)\right],
\end{numcases}
\end{subequations}
as a generalization of Eq.~(\ref{con:QSDESemiKernal}) and Eq.~(\ref{con:dUtSDEInfinite}).

The quantum stochastic dynamics in Eqs.~(\ref{con:QSDEPropagatorGeneralFormat},\ref{con:QSDEGeneralFormat}) will reduce to the traditional Markvoain scenario as  in~\citep{fischer2018particle} when the commutations $ \left[\mathfrak{b}_{\rm in}^{(j)} (t),{\mathfrak{b}_{\rm in}^{(l)}}^{\dag} (\tau)\right] \propto \delta_{jl} \delta(t-\tau) $, then the \rm{It\={o}} rules for quantum noises agree with the standard quantum stochastic calculus in~\citep{gardiner2004quantum,zhang2013non,li2022control,baragiola2012n}.

\subsection{Multi-point coupling analysis}
\label{multipoinggeneral}
When atoms are coupled to the waveguide via multiple points, the QSDE can be generalized from the single-point coupling circumstance in Sec.~\ref{generalMaster} by replacing the noise operators $\mathrm{d} \mathfrak{B}_{\rm in}^{(j)},{\mathfrak{b}}_{\rm in}^{(j)} (t)$ in Eqs.~(\ref{con:QSDEGeneralFormat},\ref{con:QSDEKernelSemiInfinite}) with $\mathrm{d} B_{\rm m,in}^{(j)},b_{\rm m,in}^{(j)}(t)$ for the semi-infinite waveguide case, and replacing $\mathrm{d} \mathfrak{B}_{\rm in}^{(j)},b_{\rm m,in}^{(j)}(t)$ with $\mathrm{d} C_{\rm m,in}^{(j)} + \mathrm{d} D_{\rm m,in}^{(j)}, c_{\rm m,in}^{(j)}(t)+d_{\rm m,in}^{(j)}(t)$ for the infinite waveguide case. Based on Theorems~\ref{ItoSemi}, \ref{ItoInfinite}, the \rm{It\={o}} rules for the quantum noise operators when atoms are coupled to the waveguide via multiple points are given by the following two theorems.

\begin{Theorem} \label{ItoSemiInfiniteMultipoint}
Based on Assumption~\ref{InitialState} and  $\mathrm{d}t< z_1^{(1)}/c$, the non-Markovian \rm{It\={o}} rule for atoms coupled to a semi-infinite waveguide via multiple points reads
\begin{equation} \label{con:NonMarkovianItoMultiPoint}
\begin{aligned} 
\mathrm{d} B_{\rm m,in}^{(j)}(t) {\mathrm{d} B_{\rm m,in}^{(l){\dag}}}(t) =  \int_{t}^{t+\mathrm{d}t}\int_{t}^{t+\mathrm{d}t}  \left[b_{\rm m,in}^{(j)}(s), b_{\rm m, in}^{(l)\dag}(\mu)\right]\mathrm{d}s\mathrm{d}\mu,
\end{aligned}
\end{equation}
where the commutator 
\begin{equation} 
\begin{aligned}
 \left[ b_{\rm m,in}^{(j)} (s), {b_{\rm m,in}^{(l)}}^{\dag} (\mu) \right] =&\int_0^s \tilde{\eta}_j^* (s-\nu) \tilde{\eta}_l (\mu-\nu) \mathrm{d}\nu,
\end{aligned}
\end{equation}
with
\begin{equation} 
\begin{aligned} 
&\tilde{\eta}_j (\mu-\nu) 
=  \sum_{n=1}^{\mathbf{n}_j} \left[ \sqrt{\gamma_{jR}^{(n)}} \delta\left(\mu- \nu-\frac{z_j^{(n)}}{c}  \right) -\sqrt{\gamma_{jL}^{(n)}} \delta \left(\mu- \nu+\frac{z_j^{(n)}}{c}  \right) \right],\notag
\end{aligned}
\end{equation}
and $\mathrm{d} B_{\rm m,in}^{(j)} \mathrm{d} B_{\rm m,in}^{(l)} = {\mathrm{d} B_{\rm m,in}^{(j){\dag}}} {\mathrm{d} B_{\rm m,in}^{(l){\dag}}} ={\mathrm{d} B_{\rm m,in}^{(l){\dag}}} \mathrm{d} B_{\rm m,in}^{(j)} = 0$.
\end{Theorem}
\begin{proof}
Eq.~(\ref{con:NonMarkovianItoMultiPoint}) is generalized from Theorem~\ref{ItoSemi}, and the kernel $\tilde{\eta}_j$ is generalized from the kernel for single-point coupling circumstance in Eq.~(\ref{con:kappajDef2}), thus more details are omitted. \qed
\end{proof}

\begin{Theorem} \label{ItoInfiniteMultipoint}
Based on Assumption~\ref{InitialStateInfiniteWave}, the non-Markovian \rm{It\={o}} rule for atoms coupled to an infinite waveguide via multiple points reads
\begin{subequations} \label{con:NonMarkovianItoInf}
\begin{numcases}{}
\mathrm{d} C_{\rm m,in}^{(j)}(t) {\mathrm{d} C_{\rm m,in}^{(l){\dag}}}(t) =\int_{t}^{t+\mathrm{d}t}\int_{t}^{t+\mathrm{d}t}  \left[c_{\rm m,in}^{(j)}(s), c_{\rm m, in}^{(l)\dag}(\mu)\right]\mathrm{d}s\mathrm{d}\mu,\\
\mathrm{d} D_{\rm m,in}^{(j)}(t) {\mathrm{d} D_{\rm m,in}^{(l){\dag}}}(t) =\int_{t}^{t+\mathrm{d}t}\int_{t}^{t+\mathrm{d}t}  \left[d_{\rm m,in}^{(j)}(s), d_{\rm m, in}^{(l)\dag}(\mu)\right]\mathrm{d}s\mathrm{d}\mu,
\end{numcases}
\end{subequations}
\begin{subequations} 
\begin{numcases}{}
\left[ c_{\rm m,in}^{(j)} (s), c_{\rm m,in}^{(l)\dag} (\mu) \right] =\int_0^s \hat{\eta}_j^* (s-\nu) \hat{\eta}_l (\mu-\nu) \mathrm{d}\nu,\\
\left[ d_{\rm m,in}^{(j)} (s), d_{\rm m,in}^{(l)\dag} (\mu) \right] =\int_0^s \breve{\eta}_j^* (s-\nu) \breve{\eta}_l (\mu-\nu) \mathrm{d}\nu,
\end{numcases}
\end{subequations}
and
\begin{subequations} \label{con:CDkernelMultipoint}
\begin{numcases}{}
\hat{\eta}_j (\mu-\nu)  = \sum_{n=1}^{\mathbf{n}_j}\sqrt{\gamma_{jR}^{(n)} }\delta\left(\mu- \nu-\frac{z_j^{(n)}}{c}  \right) ,\\
\breve{\eta}_j (\mu-\nu) =\sum_{n=1}^{\mathbf{n}_j} \sqrt{\gamma_{jL}^{(n)}} \delta\left(\mu- \nu+\frac{z_j^{(n)}}{c}  \right) .
\end{numcases}
\end{subequations}
Besides, $\mathrm{d} C_{\rm m,in}^{(j)} \mathrm{d} C_{\rm m,in}^{(l)} = {\mathrm{d} C_{\rm m,in}^{(j){\dag}}} {\mathrm{d} C_{\rm m,in}^{(l){\dag}}} ={\mathrm{d} C_{\rm m,in}^{(l){\dag}}}\mathrm{d} C_{\rm m,in}^{(j)} = 0$ and $\mathrm{d} D_{\rm m,in}^{(j)} \mathrm{d} D_{\rm m,in}^{(l)} = {\mathrm{d} D_{\rm m,in}^{(j){\dag}}} {\mathrm{d} D_{\rm m,in}^{(l){\dag}}} ={\mathrm{d} D_{\rm m,in}^{(l){\dag}}} \mathrm{d} D_{\rm m,in}^{(j)} = 0$.
\end{Theorem}
\begin{proof}
This theorem is generalized from the single-point coupling case in Theorem~\ref{ItoInfinite}, and the proof is similar to that for Theorem~\ref{ItoSemiInfiniteMultipoint}, thus more details are omitted. \qed
\end{proof}

\subsection{Simplifications with the Markovian approximation}\label{MarkovianSubsection}
In this subsection, we further clarify the simplification of the dynamics above from the perspective of input-output relationships and averaged master equations. Based on the following assumption, we clarify the properties of the quantum networks above within the Markovian approximation.
\begin{assumption} \label{Aromfrequency}
We assume that the resonant frequencies of atoms are identical as $\omega_{1}^a = \omega_{2}^a = \cdots = \omega_{N}^a = \omega_{a}$.
\end{assumption}

\subsubsection{Simplifications for input-output relationships}
Based on Assumption~\ref{Aromfrequency}, for each atom, $\sigma_j^-\left (t - t'\right) \approx e^{i\omega_a t'} \sigma_j^-\left (t\right) $~\citep{lalumiere2013input}, the non-Markovian delay-dependent input-output formalism can be represented with a linear equation with delayed phases rather than time delays. As an illustration, Eq.~(\ref{con:boutt}) can be reduced to
\begin{equation} \label{con:bouttMarkovianApp}
\begin{aligned}
b_{\rm out}'(t) =b_{\rm in}(t)  + \sum_j     \left (\sqrt{\gamma_{jL}}e^{i\omega_a z_j/c}-\sqrt{\gamma_{jR}} e^{-i\omega_a z_j/c}\right)\sigma_j^-\left (t\right),
\end{aligned}
\end{equation}
where $b_{\rm out}'(t)$ is a simplified format of $b_{\rm out}(t)$. Similarly, Eqs.~(\ref{con:bouttInFWdout},\ref{con:bouttInFWcout}) can be simplified to
\begin{subequations} \label{con:InfiniteOutputSimp}
\begin{numcases}{}
d_{\rm out}'(t) = d_{\rm in}(t)  +\sum_j    \sqrt{\gamma_{jL}} e^{i\omega_a z_j/c}\sigma_j^- \left(t \right), \label{InfdoutSim}\\
c_{\rm out}'(t) = c_{\rm in}(t)  +\sum_j    \sqrt{\gamma_{jR}} e^{-i\omega_a z_j/c}\sigma_j^- \left(t\right).\label{InfcoutSim}
\end{numcases}
\end{subequations}

On one hand, for the input-output relationship for semi-infinite waveguide in Eq.~(\ref{con:bouttMultiPoint}) with $N=1$, within the Markovian approximation, the input-output property for the multi-point coupling method can be comparable to the single-point coupling approach according to the following proposition.
\begin{mypro}  \label{MultiPointEQsinglePropo}
The difference between the input and output fields for multi-point coupling scheme between one two-level atom and a semi-infinite waveguide is identical to that of the single-point coupling scheme when   
\begin{small}
\begin{subequations} \label{con:multiPEQSingleP}
\begin{numcases}{}
\left(\sqrt{\gamma_{1L}} + \sqrt{\gamma_{1R}}\right) \sin\left( \omega_a \frac{z_1}{c}\right)= \sum_{n=1}^{\mathbf{n}_1} \left[\sqrt{\gamma_{1L}^{(n)} }+ \sqrt{\gamma_{1R}^{(n)}}\right]\sin\left[\omega_a \frac{z_1^{(n)}}{c}\right],\\
\left(\sqrt{\gamma_{1L}} - \sqrt{\gamma_{1R}}\right) \cos\left( \omega_a \frac{z_1}{c}\right) = \sum_{n=1}^{\mathbf{n}_1} \left[\sqrt{\gamma_{1L}^{(n)}} - \sqrt{\gamma_{1R}^{(n)}}\right]\cos\left[\omega_a \frac{z_1^{(n)}}{c}\right].
\end{numcases}
\end{subequations}
\end{small}%
\end{mypro}
\begin{proof}
Applying the Markovian approximation upon Eq.~(\ref{con:bouttMultiPoint}), 
\begin{equation} \label{con:bouttMultiPointN1MaAProx}
\begin{aligned}
&b_{\rm m,out}(t) \approx b_{\rm m,in}(t) + \sum_{n=1}^{\mathbf{n}_1}    \left [\sqrt{\gamma_{1L}^{(n)}} e^{i\frac{\omega_a z_1^{(n)}}{c}}-\sqrt{\gamma_{1R}^{(n)}}e^{-i\frac{\omega_a z_1^{(n)}}{c}}\right]\sigma_1^-(t).
\end{aligned}
\end{equation}
Then Eq.~(\ref{con:multiPEQSingleP}) can be derived by comparing the parameters of Eq.~(\ref{con:bouttMultiPointN1MaAProx}) with Eq.~(\ref{con:bouttMarkovianApp}). \qed
\end{proof}

On the other hand, when one two-level atom is coupled to an infinite waveguide via multi-point coupling, we have the following proposition. 
\begin{mypro} \label{MultiPointEQsinglePropoInfinite}
The difference between the input and output fields for multi-point coupling scheme between one two-level atom and an infinite waveguide is identical to that of the single-point coupling scheme when 
\begin{subequations} \label{con:multiPEQSinglePInfinite}
\begin{numcases}{}
\sqrt{\gamma_{1L}} e^{ i\omega_a z_1/c} = \sum_{n=1}^{\mathbf{n}_1} \sqrt{\gamma_{1L}^{(n)}}e^{ i\omega_a z_1^{(n)}/c},\\
\sqrt{\gamma_{1R}} e^{ -i\omega_a z_1/c} = \sum_{n=1}^{\mathbf{n}_1} \sqrt{\gamma_{1R}^{(n)}}e^{ -i\omega_a z_1^{(n)}/c}.
\end{numcases}
\end{subequations}
\end{mypro}
\begin{proof}
The proof is similar to that of Proposition~\ref{MultiPointEQsinglePropo} by comparing Eq.~(\ref{con:InfiniteOutputSimp}) with Eq.~(\ref{eq:InputOutputMultiPInfintie}) within the Markovian approximation, more details are omitted. \qed
\end{proof}

Propositions~\ref{MultiPointEQsinglePropo} and \ref{MultiPointEQsinglePropoInfinite}  show that when $e^{ i\omega_a z_1/c} = e^{ i\omega_a z_1^{(n)}/c}$ for $n=1,2,\cdots,\mathbf{n}_1$, the difference between the input and output fields when one atom is coupled to the waveguide at multiple points is equivalent to the single-point circumstance if $\sqrt{\gamma_{1L} }= \sum_{n=1}^{\mathbf{n}_1} \sqrt{\gamma_{1L}^{(n)}}$ and $\sqrt{\gamma_{1R}} = \sum_{n=1}^{\mathbf{n}_1} \sqrt{\gamma_{1R}^{(n)}}$.

\subsubsection{Master equation modeling with noise commutators}
Within the Markovian approximation, the integral process in QSDE (\ref{con:QSDEGeneralFormat}) can be simplified as
\begin{equation} \label{con:QSDEGeneralFormatMarkAppro}
\begin{aligned}
\int_0^t \mathbf{\mathfrak{L}}_{jl}(t-\tau) \mathbf{L}_l(\tau) \mathrm{d} \tau  = \int_0^t \mathbf{\mathfrak{m}}_{jl}(t-\tau) \mathrm{d} \tau\mathbf{L}_l(t)  ,
\end{aligned}
\end{equation}
where the integral kernel can be simplified as
\begin{equation} \label{con:mjlSemiInfiniteMark}
\begin{aligned}
\mathbf{\mathfrak{m}}_{jl}(t-\tau) &= \left [ {\mathfrak{b}}_{\rm in}^{(j)} (t),  \mathfrak{b}_{\rm in}^{(l)\dag} (\tau) \right]e^{i\omega_a (t-\tau)},
\end{aligned}
\end{equation}
which is determined by the quantum noise commutators derived in Sec.~\ref{Sec:NonMarknoise}. Similar for the Hermitian conjugation of Eq.~(\ref{con:QSDEGeneralFormatMarkAppro}).

We conclude  $\mathrm{d} B_{\rm m,in}^{(j)}$  in Theorem~\ref{ItoSemiInfiniteMultipoint} and $\mathrm{d} C_{\rm m,in}^{(j)} + \mathrm{d} D_{\rm m,in}^{(j)}$ in Theorem~\ref{ItoInfiniteMultipoint} as  $\mathrm{d} \mathfrak{B}_{\rm m,in}^{(j)} $, and $\mathfrak{b}_{\rm m,in}^{(j)} = \mathrm{d} \mathfrak{B}_{\rm m,in}^{(j)}/\mathrm{d}t$. Based on the non-Markovian integral QSDE (\ref{con:QSDEGeneralFormat}) with the kernel determined by quantum noises, the master equation can be derived by transforming to the Schr\"{o}dinger picture, as elaborated in detail by~\citep{zhang2013non}. Moreover, within the Markovian approximation, the master equation can be simplified as
\begin{equation} \label{con:masterquationGeneralMark}
\begin{aligned}
\dot{\rho}  = & -i \left [H,\rho  \right]  +\sum_{j,l=1}^N  \left( \int_0^t \mathbf{\mathfrak{m}}_{jl}(t-\tau)  \mathrm{d} \tau \left[L_l\rho,L_j^{\dag} \right]    + \int_0^t \mathbf{\mathfrak{m}}_{jl}^*(t-\tau)  \mathrm{d} \tau  \left[ L_j ,\rho L_l^{\dag} \right]\right),
\end{aligned}
\end{equation}
where the integral kernels with complex values are given by Eq.~(\ref{con:mjkernelMultipoint}) in the following proposition, and the single-point coupling scenario in Eq.~(\ref{con:mjlSemiInfiniteMark}) can be regarded as a simplified case with $\mathbf{n}_j = 1$. The proof of Proposition~\ref{MultiPointMasterkernel} is similar to the derivation of Eq.~(\ref{con:mjlSemiInfiniteMark}), thus is omitted.
\begin{mypro} \label{MultiPointMasterkernel}
The integral kernel in the master equation~(\ref{con:masterquationGeneralMark}) for the multi-point coupling scenario within the Markovian approximation  is 
\begin{equation} \label{con:mjkernelMultipoint}
\begin{aligned}
\mathbf{\mathfrak{m}}_{jl}(t-\tau) &= \left [\mathfrak{b}_{\rm m,in}^{(j)} (t), {\mathfrak{b}_{\rm m,in}^{(l){\dag} }}(\tau) \right]e^{i\omega_a (t-\tau)}.
\end{aligned}
\end{equation}
\end{mypro}

Considering that the real parts of integrals $\int_0^t \mathbf{\mathfrak{m}}_{jl}(t-\tau)  \mathrm{d} \tau$ are always positive according to Lemmas~\ref{lemmabinjt}-\ref{lemmadintInfinite} or the Markovian approximation upon Eqs.~(\ref{con:GammajlGauge},\ref{con:GammajlGauge2}), thus the complete positivity of the waveguide-QED system can be ensured in Eq.~(\ref{con:masterquationGeneralMark}). However, the atomic evolutions are nolonger unitary after regarding the waveguide as an environment. These properties agree with the traditional modeling approaches in \citep{pichler2015quantum,kockum2018decoherence}.

The master equation~(\ref{con:masterquationGeneralMark}) for the first time clarifies that the dynamics of a non-Markovian quantum network is determined by the commutative relationships among quantum noises. By synthesizing the QSDE for single-point coupling method in Sec.~\ref{generalMaster}, that for multi-point coupling method in Sec.~\ref{multipoinggeneral}, and the master equation (\ref{con:masterquationGeneralMark}), we conclude that non-Markovian quantum network dynamics with time delays in waveguide-QED systems can be expressed in an integral format, where the integral kernels are determined by the quantum noise commutators. This kernel-based approach generalizes the SLH approach such as in ~\citep{gough2009SLH,kockum2018decoherence,soro2022chiral,combes2017slh} or the traditional approach via averaging the environment such as in \citep{pichler2015quantum} to the circumstance that a delayed quantum network interacts with an arbitrary noisy non-Markovian environment, if only the commutators for the environmental noisy operators are clarified. In the following section, we take quantum filtering based on waveguide-QED as an example to clarify possible applications of the above theoretical method.

\section{Applications in quantum network filtering} \label{Sec:filtering}
Using the above kernel-based modeling for a multi-atom quantum network, we analyze the observation and filtering upon the quantum system in Fig.~\ref{fig:NatomWaveguide}. For the semi-infinite waveguide in Fig.~\ref{fig:NatomWaveguide}(a), the measurement information for the quantum system can be collected at the right end of the waveguide via homodyne detection. For the infinite waveguide case in Fig.~\ref{fig:NatomWaveguide}(b), the quantum network can be measured at both the left and right ends. We adopt the assumption that the measurement efficiency of homodyne detection is ideal as $1$.

We first analyze the circumstance that atoms are coupled to the waveguide via single-point coupling, then generalize to the multi-point coupling circumstance.

\subsection{Single-point coupling}
Based on the input-output relationship in Eq.~(\ref{con:bouttV2}) for semi-infinite waveguide and that in Eq.~(\ref{con:InfiniteOutputKernel}) for infinite waveguide, we conclude the  input-out formalism with kernels in a differential format as $
\mathrm{d}\mathbb{B}_{\rm out}(t) = \mathrm{d}\mathbb{B}_{\rm in}(t) +\left[ \sum_j\int_0^t \varkappa_j(t-\nu) \mathbb{L}_j(\nu)\mathrm{d}\nu\right]\mathrm{d}t$.

The measurement output quadrature $\mathbb{Y}(t)$ can be defined as~\citep{van2005quantum,chia2011quantum,gough2012quantumPRAFiltering}
\begin{equation} \label{con:quadratureYt}
\begin{aligned}
\mathbb{Y}(t) = \mathbb{B}_{\rm out}(t) + \mathbb{B}_{\rm out}^{\dag}(t),
\end{aligned}
\end{equation}
and
\begin{equation} \label{con:diffquadratureYt}
\begin{aligned}
&\mathrm{d}\mathbb{Y}(t) =\mathrm{d}\mathbb{B}_{\rm in}(t) + \mathrm{d}\mathbb{B}_{\rm in}^{\dag}(t) + \left\{\sum_j \int_0^t\left[ \varkappa_j(t-\nu) \mathbb{L}_j(\nu) +\varkappa_j^*(t-\nu) \mathbb{L}_j^{\dag}(\nu) \right]\mathrm{d}\nu\right\}\mathrm{d}t.
\end{aligned}
\end{equation}

\begin{remark} \label{Ycommutativity}
The commutativity $\left[\mathbb{Y}(s), \mathbb{Y}(t)\right] = \min(s,t) $ because $\left[\mathrm{d}\mathbb{B}_{\rm in}(t),\mathrm{d}\mathbb{B}_{\rm in}^{\dag}(s) \right] \propto \delta(t-s)$~\citep{gardiner1985input}, similar to~\citep{vladimirov2022moment}, without considering the non-Markovian interactions between the atom network and waveguide.
\end{remark}

When employing the Markovian approximation, Eq.~(\ref{con:diffquadratureYt}) can be simplified as
\begin{equation} \label{con:diffquadratureYtMarkaAp}
\begin{aligned}
\mathrm{d}\mathbb{Y}'(t) =&\mathrm{d}\mathbb{B}_{\rm in}(t) + \mathrm{d}\mathbb{B}_{\rm in}^{\dag}(t) + \sum_j \left[\varkappa_j' \mathbb{L}_j(t) + {\varkappa_j'}^* \mathbb{L}_j^{\dag}(t) \right]\mathrm{d}t,\notag
\end{aligned}
\end{equation}
where $\mathbb{Y}'$ and $\varkappa_j'$ are the simplified formats of $\mathbb{Y}$ and $\varkappa_j$ respectively, and $\varkappa_j'$ can be determined by the complex value phases in Eq.~(\ref{con:bouttMarkovianApp}) and Eq.~(\ref{con:InfiniteOutputSimp}). Then the measurement output can be influenced by the atom-waveguide interactions after integrating $\mathrm{d}\mathbb{Y}(t)$ in Eq.~(\ref{con:diffquadratureYt}) or $\mathrm{d}\mathbb{Y}'(t)$ relied on the Markovian approximation.

\subsection{Multi-point coupling} \label{multipointcouplingFilterM}
When atoms are coupled to a waveguide via multi-point coupling as in Fig.~\ref{fig:NatomWaveguide}(c), the measurement $\mathbb{Y}_{\rm m}(t)$ can be generalized from Eq.~(\ref{con:diffquadratureYt}) as
\begin{small}
\begin{equation} \label{con:diffquadratureYtMultipoint}
\begin{aligned}
\mathrm{d}\mathbb{Y}_{\rm m}(t) &=\mathrm{d}\mathbb{B}_{\rm m,in}(t) + \mathrm{d}\mathbb{B}_{\rm m,in}^{\dag}(t) + \left\{\sum_j \int_0^t \sum_{n=1}^{\mathbf{n}_j}\left[ \varkappa_j^{(n)}(t-\nu) \mathbb{L}_j(\nu)  +{\varkappa_j^{(n)}}^*(t-\nu) \mathbb{L}_j^{\dag}(\nu) \right]\mathrm{d}\nu\right\}\mathrm{d}t\\
&\triangleq \mathrm{d}\mathbb{B}_{\rm m,in}(t) + \mathrm{d}\mathbb{B}_{\rm m,in}^{\dag}(t) + \sum_j \left[\hat{\mathbb{L}}_j(t) + \hat{\mathbb{L}}_j^{\dag}(t)\right]\mathrm{d}t,
\end{aligned}
\end{equation}
\end{small}%
where $\hat{\mathbb{L}}_j(t) = \int_0^t \sum_{n=1}^{\mathbf{n}_j} \varkappa_j^{(n)}(t-\nu) \mathbb{L}_j(\nu) \mathrm{d}\nu$, the kernel $\varkappa_j^{(n)}$ can be determined by Eq.~(\ref{con:kappajDefMlultiPoint}) for semi-infinite waveguide, and Eq.~(\ref{con:kappajDefMlultiPointInfinite}) for infinite waveguide, respectively. 

\subsection{Quantum network filtering within the Markovian approximation}
Based on the general multi-point coupling scenario in Sec.~\ref{multipointcouplingFilterM}, where the single-point coupling instance can be viewed as a simplified version of multi-point coupling with $\mathbf{n}_j = 1$, the measurement output is $\mathbb{Y}_{\rm m}(t)$ in Eq.~(\ref{con:diffquadratureYtMultipoint}), and $\mathcal{Y}_t$ is the algebra generated by $\left\{\mathbb{Y}_{\rm m}(s): 0\leq s\leq t \right\}$. For an operator $\mathbb{X}(t)$, the target of quantum filtering is to determine the evolution of the operator's estimation $\mathbf{\pi}_t(\mathbb{X})$ based on the measurement $\mathcal{Y}_t$ via the least-square estimation, namely~\citep{SIAMFiltering,gough2012quantumPRAFiltering}
\[
\mathbf{\pi}_t(\mathbb{X}) = \mathbb{E} \left(\mathbb{X}(t) | \mathcal{Y}_t\right) \triangleq \rm{Tr} \left[\mathbb{X} \rho_c(t) \right],
\]
where $\mathbb{E}$ represents an estimator, and $\rho_c(t)$ represents the conditional density matrix based on measurements in quantum filtering. The error in filtering is determined by the difference between the measurement output $\mathbb{Y}_{\rm m}(t)$ and the theoretical observer values, namely $W(t)= \mathbb{Y}_{\rm m}(t) -  \sum_j \int_0^t\left[\hat{\mathbb{L}}_j(\nu) + \hat{\mathbb{L}}_j^{\dag}(\nu)\right]\mathrm{d}\nu  =\mathbb{B}_{\rm m,in}(t) + \mathbb{B}_{\rm m,in}^{\dag}(t)$ according to Eq.~(\ref{con:diffquadratureYtMultipoint}). Then we have the following remark.
\begin{remark} \label{dWWiener}
The differential of $W(t)$, denoted as $\mathrm{d} W$, is determined by the input quantum noise, and is independent from the non-Markovian atom-waveguide interactions. Then based on the widely accepted assumption that the initial vacuum noises are uniformly distributed with bandwidth~\citep{gardiner1985input} or the more stringent Assumptions~\ref{InitialState},\ref{InitialStateInfiniteWave} in this research work, $W(t)$ is a Wiener process as in~\citep{SIAMFiltering,xue2016feedback}, and $E\left(\mathrm{d} W^2\right) = \mathrm{d} t$.
\end{remark}

Then the filtering equation for general multi-point coupling circumstance can be expressed by the following proposition within the Markovian approximation. 
\begin{mypro}  \label{FilterEQ}
When the Markovian approximation is adopted for atom-waveguide interactions, the quantum filtering equation for the atom-waveguide network reads
\begin{small}
\begin{equation} \label{con:filteringEquation}
\begin{aligned}
\mathrm{d} \mathbf{\pi}_t(\mathbb{X}) = &\mathbf{\pi}_t \left\{\sum_{j,l=1}^N  \left( \left[\mathbf{L}_j^{\dag}(t),\mathbb{X} \right]\int_0^t \mathbf{\mathfrak{m}}_{jl}(t-\tau) \mathrm{d} \tau\mathbf{L}_l(t) +\int_0^t \mathbf{\mathfrak{m}}^*_{jl}(t-\tau) \mathrm{d} \tau\mathbf{L}_l^{\dag}(t) \left[\mathbb{X},\mathbf{L}_j(t) \right]\right)\right\}\mathrm{d}t \\
&+ \left[ \mathbf{\pi}_t \left(\mathbb{X} \bar{\mathbb{L}} + \bar{\mathbb{L}}^{\dag}\mathbb{X} \right) - \mathbf{\pi}_t \left(\bar{\mathbb{L}} + \bar{\mathbb{L}}^{\dag} \right)\mathbf{\pi}_t \left(\mathbb{X}\right)\right] \left[ \mathrm{d}\mathbb{Y}_{\rm m}  - \mathbf{\pi}_t \left(\bar{\mathbb{L}} + \bar{\mathbb{L}}^{\dag} \right)\mathrm{d}t\right],
\end{aligned}
\end{equation}
\end{small}%
where integral kernels $\mathbf{\mathfrak{m}}_{jl}(t-\tau)$ are given by Eq.~(\ref{con:mjlSemiInfiniteMark}) for single-point coupling and by Eq.~(\ref{con:mjkernelMultipoint}) for multi-point coupling,
\begin{equation} \label{con:EffectiveL}
\begin{aligned}
\bar{\mathbb{L}} =\sum_{j} \sum_{n=1}^{\mathbf{n}_j}   l_{jn} \mathbb{L}_j,
\end{aligned}
\end{equation}
with $l_{jn} =  \sqrt{\gamma_{jL}^{(n)} }e^{i\omega_a z_j^{(n)}/c}-\sqrt{\gamma_{jR}^{(n)}}e^{-i\omega_a z_j^{(n)}/c}$ for the semi-infinite waveguide, $l_{jn}= \sum_{n=1}^{\mathbf{n}_1} \sqrt{\gamma_{jL}^{(n)}}e^{ i\omega_a z_j^{(n)}/c}$ for the infinite waveguide measured at the left end, and $l_{jn}= \sum_{n=1}^{\mathbf{n}_1} \sqrt{\gamma_{jR}^{(n)}}e^{ -i\omega_a z_j^{(n)}/c}$ for the infinite waveguide measured at the right end.
\end{mypro}
\begin{proof}
The first part on the RHS of Eq.~(\ref{con:filteringEquation}) is free from the observer, and is the same as the non-Markovian integral Lindblad components in Eq.~(\ref{con:mjlSemiInfiniteMark}).
Eq.~(\ref{con:EffectiveL}) can be derived using the input-output relationship of the waveguide-QED network within the Markovian approximation, namely as Eq.~(\ref{con:bouttMarkovianApp}) for the semi-infinite waveguide and Eq.~(\ref{con:InfiniteOutputSimp}) for the infinite waveguide. Then combined with Remark~\ref{dWWiener}, the second part of the quantum filtering equation (\ref{con:filteringEquation}) can be derived according to~\citep{SIAMFiltering} by replacing the observing operator with $\bar{\mathbb{L}}$. \qed
\end{proof}

Generalized from Proposition~\ref{FilterEQ}, the filtering equation can be equivalently represented with the density matrix as
\begin{equation} \label{eq:measFBmaintext}
\begin{aligned}
&\mathrm{d}\rho_F=  -i \left[H,\rho_F\right]\mathrm{d}t+ \mathcal{H}\left[\bar{\mathbb{L}}\right]\rho_F \mathrm{d}W +\sum_{j,l=1}^N  \left( \int_0^t \mathbf{\mathfrak{m}}_{jl}(t-\tau) \mathrm{d} \tau \left[L_l\rho_F,L_j^{\dag} \right]    + \int_0^t \mathbf{\mathfrak{m}}_{jl}^*(t-\tau)  \mathrm{d} \tau  \left[ L_j ,\rho_F L_l^{\dag} \right]\right) \mathrm{d}t, 
\end{aligned}
\end{equation}
where $\bar{\mathbb{L}}$ is given by Eq.~(\ref{con:EffectiveL}), the superoperator reads
\[
 \mathcal{H}[\bar{\mathbb{L}}]\rho_F = \bar{\mathbb{L}} \rho_F + \rho_F \bar{\mathbb{L}}^{\dag} -  \mathrm{Tr} \left [  \right ( \bar{\mathbb{L}} +\bar{\mathbb{L}}^{\dag}\left )\rho_F \right ]\rho_F.
\]
This generalizes the filtering equation for quantum systems in \citep{wiseman2009quantum,wiseman1994quantum,zhang2017quantum,ding2023quantumSIAM} to the atom network with multiple delayed integral kernels.

\subsection{Numerical simulations}
\subsubsection{Two atoms coupled to a semi-infinite waveguide}
When there are two entangled atoms coupled to a semi-infinite waveguide via one-point coupling~\citep{ZhangBin}, the time-varied density matrix reads
\begin{small}
\begin{equation} \label{con:AsMcavity}
\begin{aligned}
\rho_2(t) &= \begin{bmatrix}
0  & 0  & 0 & 0\\
0 &|\alpha(t)|^2 & \alpha^*(t)\beta(t) & 0\\
0 & \alpha(t)\beta^*(t) & |\beta(t)|^2  & 0\\
0 & 0 &  0 &1-|\alpha(t)|^2- |\beta(t)|^2 
\end{bmatrix},
\end{aligned}
\end{equation}
\end{small}%
where $\alpha(t)$ and $\beta(t)$ represent the amplitudes that only the first or the second atom is excited, respectively.

The filtering dynamics of $\rho_2 (t)$ is governed by Eq.~(\ref{eq:measFBmaintext}) with $N=2$. We assume that initially the two atoms are entangled and there are no applied drives,   $\mathrm{d}t = 0.5\mu s$, $\gamma_{1R} = \gamma_{1L} = 0.5\gamma_{2R} = 0.5\gamma_{2L} = 0.2$MHz, and initially the first atom is excited~\citep{verdu2009strong,kockum2018decoherence,2017CouplingNC,kannan2020waveguide}. The Hamiltonian representing resonant drives with amplitudes  $\Omega_1$ and $\Omega_2$ applied upon two atoms is $H =\Omega_1 \begin{bmatrix}  0  & 1 \\ 1 &  0\end{bmatrix} \otimes \begin{bmatrix}  1  & 0 \\ 0 &  1\end{bmatrix}  +\Omega_2 \begin{bmatrix}  1  & 0 \\ 0 &  1\end{bmatrix} \otimes \begin{bmatrix}  0  & 1 \\ 1 &  0\end{bmatrix}$. Then the convergence of quantum filtering with different coupling strengths is compared in Fig.~\ref{fig:TwoatomFiltering}, where the colored lines represent the mean dynamics for quantum filtering.  We denote $\Delta_z = \left(z_2-z_1\right)/c$, and $c$ is the velocity of the filed in the waveguide. In Fig.~\ref{fig:TwoatomFiltering}(a), $\Omega_1 = 0.1$MHz, $\Omega_2 = 0$MHz, $\omega_az_1/c = 0.3\pi$, and $\omega_az_2/c = 1.3\pi$. In Fig.~\ref{fig:TwoatomFiltering}(b), $\Omega_1 = 0$MHz, $\Omega_2 = 0.2$MHz, $\omega_az_1/c = 0.3\pi$, and $\omega_az_2/c = 0.8\pi$.  It can be seen that the evolution and filtering of quantum states can be influenced by the drive fields applied upon atoms and the non-Markovian couplings between atoms and waveguide related to atoms' positions, resulting in the oscillating and converging dynamics in Fig.~\ref{fig:TwoatomFiltering}(a) and Fig.~\ref{fig:TwoatomFiltering}(b) respectively. 
\begin{figure}[htbp]
  \centering
  \centerline{\includegraphics[width = 1\columnwidth]{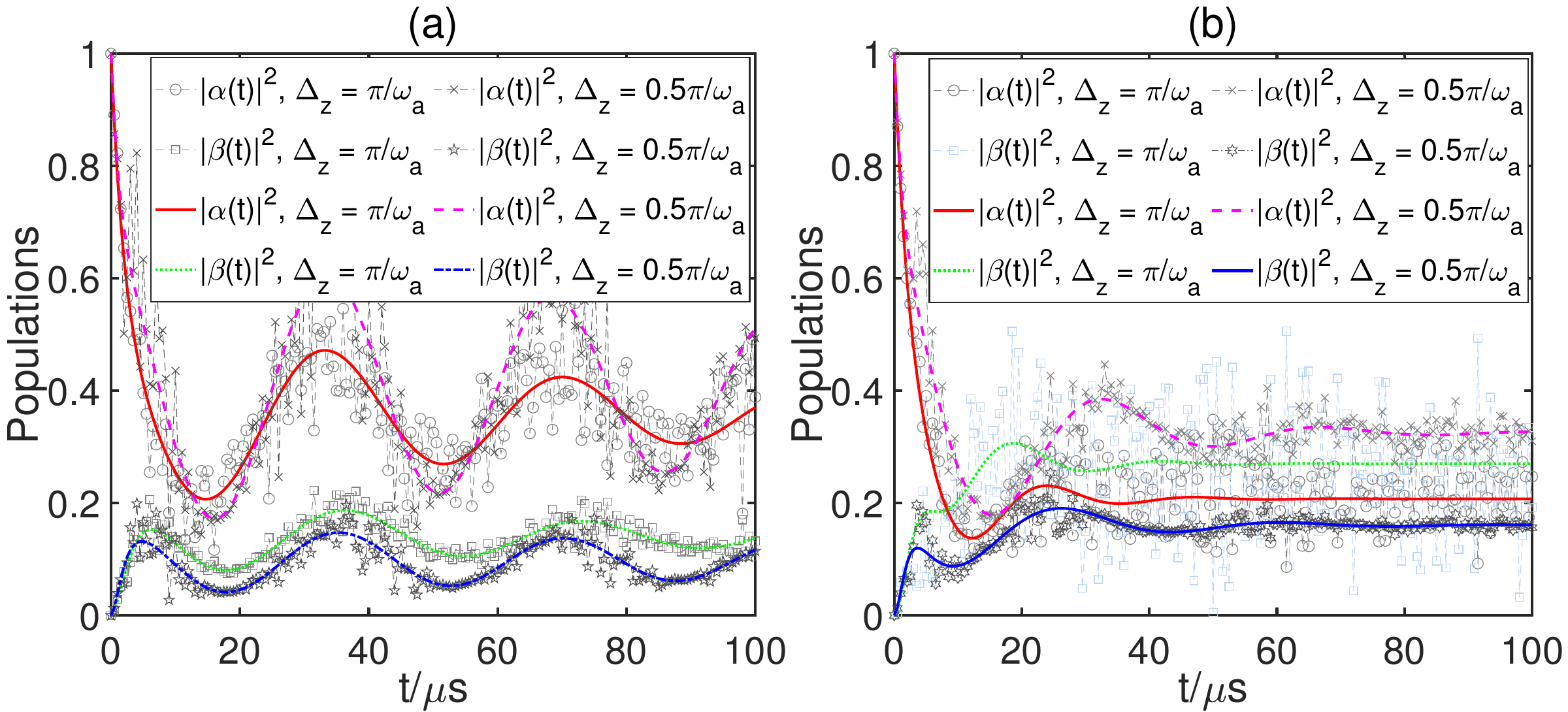}}
  \caption{Comparisons of quantum filtering evaluated by $|\alpha(t)|^2$ and $|\beta(t)|^2$ when two atoms are coupled to a semi-infinite waveguide.}
  \label{fig:TwoatomFiltering}
\end{figure}

\subsubsection{Three atoms coupled to an infinite waveguide}
Generalized from the above two-atom network, we further consider that three atoms are coupled to an infinite waveguide. In this scenario, we take a simplified case that the quantum network can be measured at the right output end as in Fig.~\ref{fig:NatomWaveguide}(b). Then the dynamics can be represented as
\begin{small}
\begin{equation} \label{con:MasterEquationSingledirection}
\begin{aligned}
&\mathrm{d}\rho_3  =  - i \left[ H_0 ,\rho_3\right]\mathrm{d}t +\sum_{j=1}^3\Gamma_j  \mathcal{D}\left[\sigma_j^-\right]\rho_3\mathrm{d}t + \mathcal{H}\left[\bar{\mathbb{L}}\right]\rho_3 \mathrm{d}W  +\left[\sum_{1\leq j,l \leq 3, j\neq l}\tilde{\Gamma}_{jl} \left( \sigma_j^- \rho_3 \sigma_l^+ - \frac{1}{2}\sigma_j^- \sigma_l^+ \rho_3 - \frac{1}{2}\rho_3 \sigma_j^- \sigma_l^+  \right) \right]\mathrm{d}t,
\end{aligned}
\end{equation} 
\end{small}%
where $\mathcal{D}[O]\rho_3 = O\rho_3 O^{\dag} - \left(O^{\dag}O\rho_3 + \rho_3 O^{\dag}O \right)/2$ for an arbitrary operator $O$, $\Gamma_j = \gamma_{jR} +  \gamma_{jL} + \eta_j$ according to the Hamiltonian in Eq.~(\ref{con:freeHam0}) with non-zero loss rates, $\tilde{\Gamma}_{jl} =\left(\sqrt{\gamma_{jR}\gamma_{lR}}+ \sqrt{\gamma_{jL}\gamma_{lL}}\right)\cos \left[\omega_a \left(z_l-z_j\right)/c\right]$,
\begin{small}
\begin{equation} \label{con:H0Singledirection}
\begin{aligned}
&H_0  =\Omega_2 \left(\sigma_2^- + \sigma_2^+\right)  + \sum_{1\leq j<l \leq 3} \left[ \frac{\sqrt{\gamma_{jR}\gamma_{lR}} e^{i\omega_a \left(z_l-z_j\right)/c} -\sqrt{\gamma_{jL}\gamma_{lL}} e^{i\omega_a \left(z_j-z_l\right)/c} }{2i} \sigma_j^-\sigma_l^+ + {\rm H.c.}\right] ,
\end{aligned}
\end{equation}
\end{small}%
and $\bar{\mathbb{L}}$ is given by Eq.~(\ref{con:EffectiveL}) with $\mathbf{n}_j = 1$ for $j =1,2,3$.

In the simulations in Fig.~\ref{fig:ThreeatomFiltering}, we assume that initially only the first atom is excited, and the other atoms are at their ground states. We take $ \mathrm{d}t = 0.5\mu s$, $\gamma_{1L} =\gamma_{1R} =0.1$MHz, $\gamma_{2L} =\gamma_{2R} = 0.2$MHz, $\gamma_{3L} =\gamma_{3R} = 0.3$MHz, $z_2-z_1 = z_3-z_2 = \pi c/\omega_a $, $\Omega_2 = 0.5$MHz, $\eta_j = 0$MHz in Fig.~\ref{fig:ThreeatomFiltering}(a) and $\eta_j = 0.2$MHz in Fig.~\ref{fig:ThreeatomFiltering}(b)
~\citep{verdu2009strong,kockum2018decoherence,2017CouplingNC,kannan2020waveguide}.   
As shown in Fig.~\ref{fig:ThreeatomFiltering}(a) and Fig.~\ref{fig:ThreeatomFiltering}(b), the third atom can become excited with $\left\langle \sigma_3^z\right\rangle > -1$ by absorbing the photon emitted by the first two atoms, the second atom can be highly excited because of the applied drive, while the atoms are less excited when $\eta_j>0$. Thus atoms' dynamics can be influenced by the amplitude of the drive field $\Omega_2$, the non-Markovian interactions between atoms and waveguide, and the decay rates of the atoms $\eta_j$ to the environment outside of the waveguide.
\begin{figure}[htbp]
  \centering
  \centerline{\includegraphics[width = 1\columnwidth]{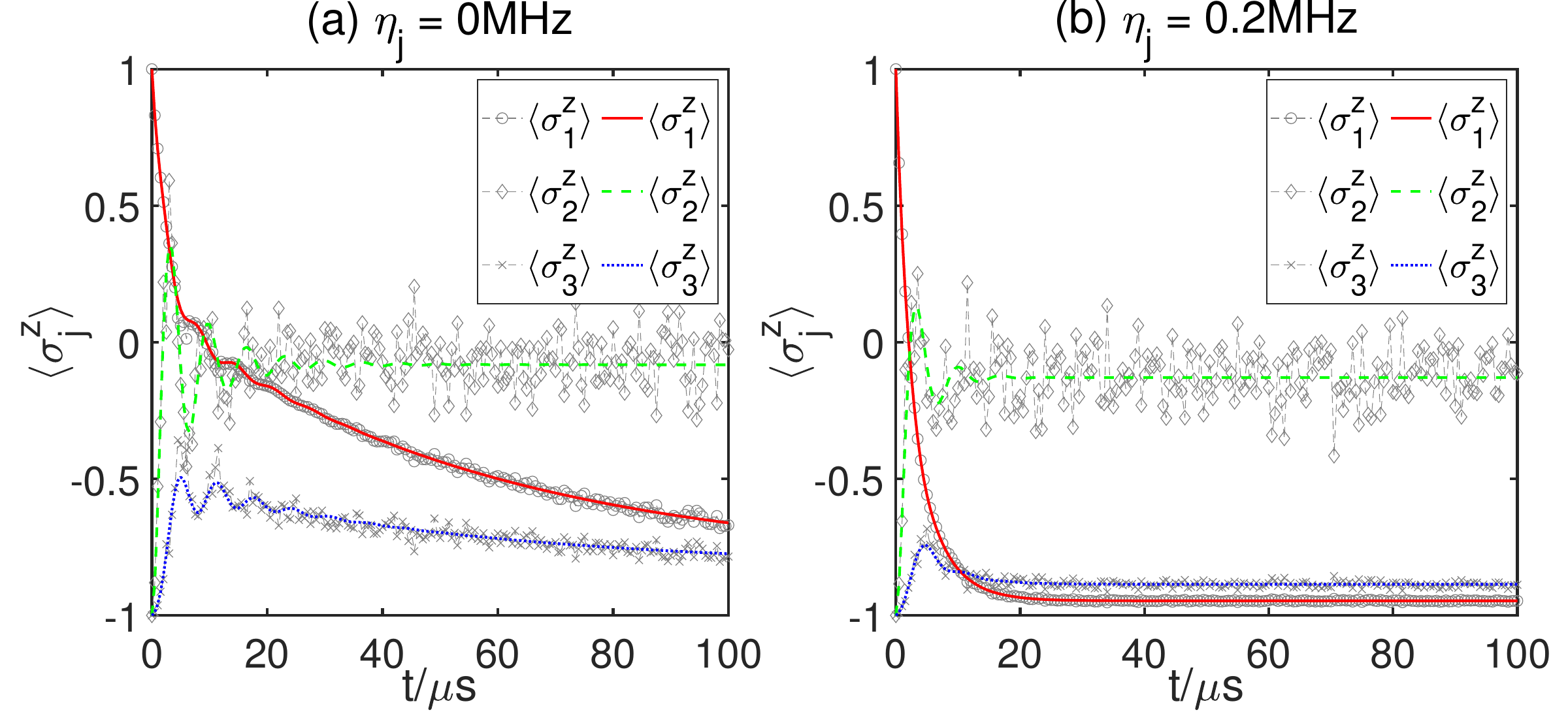}}
  \caption{Quantum filtering for a three-atom network.}
  \label{fig:ThreeatomFiltering}
\end{figure}

\section{Conclusion}\label{Sec:conclusion}
In this paper, we study the quantum stochastic dynamics of non-Markovian quantum networks with time delays realized by waveguide-QED systems from the perspective of quantum noises. The non-Markovian interactions among atoms connected by the waveguide can be described by the input-output relationship incorporating time delays. We propose that, the non-Markovian characteristics of the waveguide-QED networks can be evaluated by the non-Markovian commutative properties and \rm{It\={o}} rules of quantum noises, distinguishing them from conventional Markovian quantum networks. Consequently, the dynamical behavior of waveguide-QED systems can be equivalently expressed through a non-Markovian quantum stochastic differential equation involving integral processes, where the integral kernels are characterized by the commutators of quantum noises. This kernel-based approach can be extended to quantum networks where atoms are coupled to the waveguide via multiple points. Furthermore, the integral kernels, determined by quantum noises, enable the evaluation of measurement outcomes at the waveguide's terminus, facilitating their applications in quantum filtering based on the waveguide-QED networks.

\appendix
$$\textbf{Appendix}$$
\section{Derivation of Eq.~(\ref{con:HRWProduct})} \label{Sec:AppdixTraceWaveguide}
Based on Eq.~(\ref{con:Hspicture}), we can derive its equivalent format by rotating with the waveguide Hamiltonian as
\begin{equation} \label{con:HeffSemi0}
\begin{aligned}
\mathbf{H}(t) 
 &=\sum_{j=1}^N \left(\omega_j^a - i\frac{\eta_j}{2}\right) \sigma_j^+\sigma_j^- + \sum_{j=1}^N e^{i H_{\omega} t}  H_{I}^{(j)}  e^{-i H_{\omega} t},
\end{aligned}
\end{equation}
where the second component reads
\begin{equation}
\begin{aligned}  \label{con:MathFormuladomega}
e^{i H_{\omega} t}  H_{I}^{(j)}  e^{-i H_{\omega} t} = \int_{-\infty}^{\infty}  \left [ g_{j}(\omega,z_j)e^{i\omega t}d^{\dag}_{\omega} \sigma_j^- +  g_{j}^*(\omega,z_j) e^{-i\omega t} d_{\omega}\sigma_j^+ \right ]\mathrm{d}\omega,
\end{aligned}
\end{equation}
because $e^{i H_{\omega} t}  d^{\dag}_{\omega}  e^{-i H_{\omega} t} =  d^{\dag}_{\omega} + \left(i t \right) \left[H_{\omega} , d^{\dag}_{\omega} \right] + \frac{\left(i t \right)^2}{2!} \left[ H_{\omega} ,\left[H_{\omega} , d^{\dag}_{\omega} \right] \right] + \cdots=d^{\dag}_{\omega} +i \omega t   d^{\dag}_{\omega} + \frac{\left(i\omega t \right)^2}{2!}  d^{\dag}_{\omega} + \cdots = e^{i\omega t}d^{\dag}_{\omega}$. According to Eq.~(\ref{eq:FreqOperator}), we denote $d_{\omega} =  d_{\omega}(\omega,t_0)$ and $d_{\omega}^{\dag} =  d_{\omega}^{\dag}(\omega,t_0)$ with $t_0 = 0$. For example, the second component in the last line of Eq.~(\ref{con:MathFormuladomega}) reads
\begin{equation}
\begin{aligned}  \label{con:QSDEInput}
\int_{-\infty}^{\infty}  g_{j}^*(\omega,z_j) e^{-i\omega t} d_{\omega}\sigma_j^+ \mathrm{d}\omega &=-i\int_{-\infty}^{\infty}  \left (\sqrt{\frac{\gamma_{jR}}{2\pi}}e^{i\omega z_j/c} -\sqrt{\frac{\gamma_{jL}}{2\pi}} e^{-i\omega z_j/c}\right) e^{-i\omega t}   d_{\omega}(\omega,t_0)\sigma_j^+ \mathrm{d}\omega \\
& =i\left[  \sqrt{\gamma_{jL}} b_{\rm in}\left(t+ z_j/c\right) - \sqrt{ \gamma_{jR}} b_{\rm in}\left(t-z_j/c\right) \right]\sigma_j^+ .
\end{aligned}
\end{equation}
Then combined with its Hermitian conjugation, we can derive Eq.~(\ref{con:HRWProduct}) in the main text.

\section{Derivations of non-Markovian dynamics for the propagator $U(t)$ in Eq.~(\ref{con:dUtSDE2})} \label{Sec:AppendixPropogator}
The influence on the evolution of quantum states by quantum noises in the waveguide can be evaluated with the following commutator~\citep{gough2017series}
\begin{equation} \label{con:binUt}
\begin{aligned}
\int_0^t  \left [ b_{\rm in}^{(j)} (t), \dot{U}(\tau)\right]\mathrm{d} \tau  = &\int_0^t  \left [ b_{\rm in}^{(j)} (t),  \sum_{l=1}^N \left(  {b_{\rm in}^{(l)}}^{\dag} (\tau) L_l -\rm H.c.\right)U(\tau)\right]\mathrm{d} \tau\\
 =&\sum_{l=1}^N \int_0^t   \left [ b_{\rm in}^{(j)} (t),  {b_{\rm in}^{(l)}}^{\dag} (\tau) \right]L_lU(\tau)\mathrm{d} \tau.
\end{aligned}
\end{equation}

Based on Lemma~\ref{lemmabinjt} and  Lemma~\ref{lemmabinjtTwoatom}, for the case that $j < l$ in Eq.~(\ref{con:binUt}), 
\begin{equation} \label{con:binUtCal2}
\begin{aligned}
& \int_0^t  \left [ b_{\rm in}^{(j)} (t), {b_{\rm in}^{(l)}}^{\dag} (\tau) \right]L_l U(\tau)\mathrm{d} \tau\\
=&\int_0^t  \left [ \sqrt{\gamma_{jL}\gamma_{lL}}\delta\left( t - \tau+ \frac{z_j-z_l}{c}\right) +\sqrt{\gamma_{jR}\gamma_{lR}}\delta\left( t - \tau + \frac{z_l-z_j}{c}\right) \right.\\
&\left.-\sqrt{\gamma_{jL}\gamma_{lR} }\delta\left( t - \tau +\frac{z_j+z_l}{c}\right) -\sqrt{\gamma_{jR}\gamma_{lL}} \delta\left( t -\tau -\frac{z_j+z_l}{c}\right) \right]L_l U(\tau)\mathrm{d} \tau \\
=& \sqrt{\gamma_{jL}\gamma_{lL}} L_l U \left(t- \frac{z_l-z_j}{c} \right) - \sqrt{\gamma_{jR}\gamma_{lL}} L_l U \left(t- \frac{z_l+z_j}{c} \right).
\end{aligned}
\end{equation}
Similarly, when  $j > l$, 
\begin{equation} \label{con:binUtCal3}
\begin{aligned}
&\int_0^t  \left [ b_{\rm in}^{(j)} (t), \left(  {b_{\rm in}^{(l)}}^{\dag} (\tau) L_l -b_{\rm in}^{(l)} (\tau)L_l^{\dag}\right)U(\tau)\right]\mathrm{d} \tau = \sqrt{\gamma_{jR}\gamma_{lR}}L_l U \left(t- \frac{z_j-z_l}{c} \right) - \sqrt{\gamma_{jR}\gamma_{lL}}L_lU \left(t- \frac{z_l+z_j}{c} \right),
\end{aligned}
\end{equation}
and when $l=j$,
\begin{equation} \label{con:binUtCal1}
\begin{aligned}
&\int_0^t  \left [ b_{\rm in}^{(j)} (t),  {b_{\rm in}^{(j)}}^{\dag} (\tau)\right] L_j U(\tau)\mathrm{d} \tau =\frac{\gamma_{jR} + \gamma_{jL}}{2}L_j U(t)-\sqrt{\gamma_{jR}\gamma_{jL}}L_j U\left(t-\frac{2z_j}{c}\right).
\end{aligned}
\end{equation}

Then Eq.~(\ref{con:Uequation}) can be rewritten as Eq.~(\ref{con:dUtSDE2}) in the main text.

\section*{Acknowledgement}
\begin{sloppypar}
The authors thank Re-Bing Wu for helpful reading and discussions. This work is partially financially supported by Quantum Science and Technology-National Science and Technology Major Project 2023ZD0300600, Guangdong Provincial Quantum Science Strategic Initiative No. GDZX2303007, Hong Kong Research Grant Council (RGC) under Grant No. 15213924, and  the CAS AMSS-PolyU Joint Laboratory of Applied Mathematics.
\end{sloppypar}

\bibliographystyle{elsarticle-harv} 
\bibliography{waveguideAutomatica}

\end{document}